\documentclass[12pt]{article}
\usepackage{amsfonts,epsfig,latexsym,amscd,amsmath,theorem,mathrsfs}
\usepackage{rotating}
\usepackage{subcaption} 
\usepackage{graphicx} 
\usepackage[percent]{overpic}

\DeclareMathOperator{\Vol}{Vol}

\newcommand{\pr}{{\rm pr}}

\newcommand{\R}{{\mathbb{R}}}

\newcommand{\C}{{\mathbb{C}}}

\newcommand{\Ll}{{\mathscr{L}}}

\newcommand{\beq}{\begin{equation}}
\newcommand{\eeq}{\end{equation}}
\newcommand{\bea}{\begin{eqnarray}}
\newcommand{\eea}{\end{eqnarray}}
\newcommand{\ben}{\begin{eqnarray*}}
\newcommand{\een}{\end{eqnarray*}}
\newcommand{\bem}{\begin{enumerate}}
\newcommand{\eem}{\end{enumerate}}

\newcommand{\ra}{\rightarrow}

\newcommand{\hra}{\hookrightarrow}

\newcommand{\cd}{\partial}
\newcommand{\wt}{\widetilde}

\newcommand{\less}{\backslash}
\newcommand{\M}{{\sf M}}
\newcommand{\csob}{C_{Sob}}
\newcommand{\cvol}{C_{vol}}

\def \d{\mathrm{d}}
\newcommand{\dstar}{\delta}
\newcommand{\sym}[1]{{\rm Sym}^{#1}}
\newcommand{\ip}[1]{\langle #1 \rangle}
\newcommand{\ignore}[1]{}

\renewcommand{\star}{*}

\renewcommand{\L}{\mathsf{L}}

\newcommand{\ol}{\overline}

\newcommand{\Conn}{{\mathfrak A}}

\newcommand{\eps}{\varepsilon}
\renewcommand{\phi}{\varphi}

\theoremstyle{plain}
\newtheorem{thm}{Theorem}
\newtheorem{lemma}[thm]{Lemma}

{\theorembodyfont{\rmfamily}

\newtheorem{remark}[thm]{Remark}

}

\newcommand{\news}{\setcounter{equation}{0}}
\newenvironment{proof}{\noindent{\it Proof:\, }}{\hfill$\Box$\vspace*{0.5cm}
}

\begin{document}

\title{Chern-Simons deformation of vortices on compact domains}
\author{
S.P. Flood\thanks{Email: {\tt s.p.flood@leeds.ac.uk}}\:  and J.M. Speight\thanks{E-mail: {\tt speight@maths.leeds.ac.uk}}\\
School of Mathematics, University of Leeds\\
Leeds LS2 9JT, England}

\maketitle

\begin{abstract}
Existence of Maxwell-Chern-Simons-Higgs (MCSH) vortices in a Hermitian line bundle 
$\L$ over a general compact Riemann surface
$\Sigma$
is proved by a continuation method. The solutions are proved to be smooth both spatially and
as functions of the Chern-Simons deformation parameter $\kappa$, and exist for all $|\kappa|<\kappa_*$,
where $\kappa_*$ depends, in principle, on the geometry of $\Sigma$, the degree $n$ of $\L$, which may be interpreted as the vortex number, and the vortex positions. 
A simple upper bound on $\kappa_*$, depending only on $n$ and the volume of
$\Sigma$, is found. Further, it is proved that a positive {\em lower} bound on $\kappa_*$, depending
on $\Sigma$ and $n$, but independent of vortex positions, exists. A detailed numerical
study of rotationally equivariant vortices on round two-spheres is performed. We find that $\kappa_*$ in general does depend on vortex positions, and, for fixed $n$ and radius, tends to be larger the more evenly vortices are distributed between the North and South poles. A generalization of the MCSH model
to compact K\"ahler domains $\Sigma$ of complex dimension $k\geq 1$ is formulated. The Chern-Simons term is replaced by the integral over spacetime of $A\wedge F\wedge \omega^{k-1}$, where $\omega$ is the K\"ahler form on $\Sigma$. A topological lower bound on energy is found, attained by solutions
of a deformed version of the usual vortex equations on $\Sigma$. Existence, uniqueness
and smoothness of vortex solutions of these generalized equations is proved, for $|\kappa|<\kappa_*$, and
an upper bound on $\kappa_*$ depending only on the K\"ahler class of $\Sigma$ and the first Chern 
class of $\L$ is obtained.
\end{abstract}

\maketitle

\section{Introduction}
\news

Vortices are the simplest class of topological solitons occuring in gauge theory. Being simple, they are useful prototypes for more complicated, higher-dimensional solitons (monopoles, instantons, calorons), as well as having interesting applications in their own right, in condensed matter physics and
cosmology. They arise in the abelian Higgs model, a $(2+1)$-dimensional field theory consisting of a complex scalar field $\phi$ (the Higgs field) minimally coupled to a $U(1)$ gauge field $A$, obeying Maxwell electrodynamics. For a particular choice of the Higgs self-interaction potential, this theory exhibits the mathematically interesting property of ``self duality": there is a topological lower bound on energy which is attained by solutions of a coupled system of first order PDEs. The space of gauge equivalence classes of solutions of this system, in a fixed topological
class, is a finite dimensional smooth manifold, the so called $n$-vortex moduli space $M_n$, which inherits a canonical K\"ahler structure. One may identify $M_n$ with the space of unordered $n$-tuples of marked points in physical space, these being the points at which the Higgs field vanishes. This is true whether physical space
is $\R^2$ \cite{tau} or a compact Riemann surface \cite{bra,garpra}.
There is a well-developed formalism for extracting the low energy dynamics of vortices from the geometry of $M_n$, originally developed by Manton, see \cite{mansut} for a thorough review. 

Many elaborations on the basic abelian Higgs model preserving a self-duality structure are possible. (One can, for example, allow both physical space and the target space of the Higgs field to be
K\"ahler manifolds, and the gauge group to be any Lie group with a Hamiltonian and isometric action on target space.)
From a physical standpoint, perhaps the most interesting elaboration is the inclusion of a Chern-Simons term in the theory. This converts the vortices into dyons, that is, particles carrying both magnetic flux and electric charge, and allows the possibility of exotic exchange statistics once the theory is quantized. There are two ways to introduce a Chern-Simons term into the theory while keeping Lorentz covariance and a self-duality structure. In one \cite{jacleewei} the Maxwell term for $A$ is directly replaced by the Chern-Simons term, and the usual quartic Higgs potential is replaced by a certain sextic potential. This Chern-Simons-Higgs (CSH) model has been quite thoroughly studied but, even so, the existence theory for vortices is less well developed than for standard vortices. It is known that an $n$-vortex exists for each choice of
$n$ points in physical space $\Sigma$ if $\Sigma=\R^2$ \cite[pp.\ 164-177]{yan} or a flat torus \cite{cafyan}. Once the model is put on a compact domain, the coupling constant in front of the Chern-Simons term, usually denoted $\kappa$, becomes a nontrivial parameter (on $\R^2$ it can be scaled away). It is proved in \cite{cafyan} that for each set $D$ of $n$ marked points on a torus, there exists $\kappa_*(D)>0$, depending on $D$, such that, for all $\kappa\in(0,\kappa_*(D))$ there is a vortex solution with
$\phi^{-1}(0)=D$, and that $\kappa_*(D)$ is finite for all $D$ (that is, for large enough $\kappa$, no $n$-vortex exists). Existence of vortices on compact Riemann surfaces of higher genus has not been established, and there does not appear to be a quick and simple resolution for this. In particular, a direct application of Bradlow's approach \cite{bra} to this vortex system is unhelpful because the higher nonlinearity of the sextic Higgs potential produces an elliptic PDE with analytically difficult nonlinear terms. Even recent studies of this type of vortex, dealing with the generalization to nonabelian gauge groups, restrict themselves to the case of flat tori \cite{hanlinyan}.

In this paper, we address the vortex existence question in the second, rather less well-studied Chern-Simons vortex system, making progress on arbitrary compact domains. This model, originally due to Lee, Lee and Min \cite{leeleemin} is a one-parameter {\em deformation} of the basic abelian Higgs model, the deformation
parameter being the Chern-Simons coupling $\kappa$. The model keeps the usual Maxwell term for $A$, but adds ($\kappa$ times) the Chern-Simons term, and couples the Higgs field $\phi$ to a new neutral scalar field $N$ via a $\kappa$-dependent, but still quartic, interaction potential. We shall refer to it as the Maxwell-Chern-Simons-Higgs (MCSH) model. Following Bradlow \cite{bra}, one can, for a fixed set of vortex positions, formulate the vortex equations as a coupled {\em pair} of semilinear second order elliptic PDEs, for $|\phi|$ and $N$. Proving existence of solutions of systems of semilinear PDEs is, in general, a much more difficult problem than for a single PDE.\, Considerable progress has been made, in the case $\Sigma=T^2$ by Ricciardi and Tarantello \cite{rictar}. By a thorough analysis of the coupled system, they establish that, for each set $D$ of vortex positions, and for all $\kappa>0$ sufficiently small, there are in fact at least {\em two} inequivalent vortex solutions with $\phi^{-1}(0)=D$. They also find a global upper bound on $\kappa_*(D)$, depending only on the volume of $\Sigma$, and prove that, in two different limits, vortex solutions of the MCSH model converge to solutions of both the CSH and original abelian Higgs models. 

It is plausible that the methods of \cite{rictar} should extend to arbitrary compact 
Riemann surfaces $\Sigma$ (and, indeed, existence results for MCSH vortices on general $\Sigma$ are sometimes asserted as folk theorems on this basis \cite{hankim}). That is not, however, the aim of the current paper. Rather, we will directly exploit the deformation character of the MCSH system to give a much more elementary existence (and local uniqueness) proof of those vortices which continue smoothly to $\kappa=0$. The idea is that, at  $\kappa=0$, we know a unique solution exists for each choice of $D$, namely the standard abelian Higgs vortex augmented by $N\equiv 0$. An Implicit Function Theorem argument then allows us to deduce that, for each $D$, there is
$\kappa_*(D)>0$ such that, for all $\kappa\in(-\kappa_*(D),\kappa_*(D))$ there is a locally unique vortex solution with $\phi^{-1}(0)=D$. This solution is smooth, and depends smoothly on the deformation parameter $\kappa$. 
We also prove the existence of a positive lower bound $\kappa_{**}$ on $\kappa_*(D)$, depending on
$n$ and $\Sigma$, but independent of $D$. Hence, for all position sets $D$ of
size $n$, locally unique vortices with $\phi^{-1}(0)=D$ exist for all $-\kappa_{**}<\kappa<\kappa_{**}$. Loosely, this shows that, for sufficiently small $\kappa$, the entire moduli space of $n$-vortices $M_n$ survives the Chern-Simons deformation, a key underlying assumption of the various proposals for moduli space approximations to low energy vortex dynamics in this model \cite{kimlee,colton,alqspe}.
As far as we are aware, this is the first time existence of a global {\em lower} bound on $\kappa_*$
has been established, and the smooth continuation viewpoint is crucial to our argument. In comparison with \cite{rictar}, we obtain more refined information (existence of {\em smooth curves} of solutions parametrized by $\kappa$, and a lower bound on $\kappa_*$) in more general geometries, but only for one type of vortex: those
continuously connected to ordinary abelian Higgs vortices. Our argument is also considerably easier, using only basic facts from functional analysis. 

We will also 
find
a global {\em upper} bound on $\kappa_*(D)$, independent of $D$. 
In the case $\Sigma=T^2$, this is larger (hence worse) than a bound obtained in 
\cite{rictar}, but, again, the proof is much simpler. Neither bound is expected to be sharp. 
Our bound may be thought of as the MCSH analogue of the Bradlow bound for existence of undeformed vortices
\cite{bra}, which states that $n$-vortices cannot exist if the volume of $\Sigma$ is less than $4\pi n$. As we will see, Chern-Simons deformation makes this requirement more stringent: $\kappa$-deformed
vortices cannot exist if $\Vol(\Sigma)<4\pi n(1+\kappa^2)$.

The question arises whether the maximal coupling $\kappa_*(D)$ at which vortices with $\phi^{-1}(0)=D$ exist actually depends nontrivially on $D$. We will produce robust numerical evidence that it does, by studying the MCSH model on the round sphere (of radius $R\geq \sqrt{n}$) in the cases where $D$ consists of the north pole, with multiplicity $n_+$, and the south pole with multiplicity $n_-=n-n_+$. In these cases, the Bogomol'nyi equations reduce, due to rotational invariance, to an ODE system, which we solve numerically via a shooting method. We find that, for given $n$ and $R$, the maximal $\kappa$ at which  vortices exist depends on $n_-$, the number placed at the south pole. For example, $(n_+,n_-)=(1,1)$ vortices have larger $\kappa_*$ than $(n_+,n_-)=(2,0)$ vortices.
We also find that the solution curves, rather than disappearing at $\kappa=\kappa_*$, have a turning point in $\kappa$, continuing smoothly back towards $\kappa=0$, approaching
a singular limit in which the magnetic field becomes uniform, the Higgs field vanishes, and the
neutral scalar field becomes uniform and diverges to $-\infty$. This indicates that the two distinct vortex solutions whose existence for each $D$ (on $T^2$) and small $\kappa>0$ was proved in \cite{rictar}, merge at $\kappa=\kappa_*(D)$, and can actually be considered as a single, connected solution branch. 

To illustrate the power and elegance of the continuation/IFT
strategy, we go on to formulate a Chern-Simons-like deformation of vortices on a general
compact K\"ahler manifold of complex dimension $k$. In this setting, $D$ is an effective divisor in $\Sigma$, and the Chern-Simons term is replaced by a constant multiple of
$$
\int_{I\times\Sigma} A\wedge F\wedge \omega^{k-1}
$$
where $\omega$ is the K\"ahler form on $\Sigma$. Again, a simple IFT argument establishes the existence
and local uniqueness of a smooth curve (parametrized by $\kappa\in(-\kappa_*(D),\kappa_*(D))$) of smooth solutions for each divisor $D$.  A uniform upper bound on $\kappa_*$, depending only on
the K\"ahler class of $\Sigma$ and the first Chern class of the line bundle supporting the Higgs field, is readily found.

\section{The Maxwell-Chern-Simons-Higgs model on a general surface}\news
\label{sec:MCSH}

Let $(\Sigma,g)$ be a compact oriented Riemannian two-manifold, $(\L,h)$ a hermitian line bundle over $\Sigma$ of degree $n\geq 1$, and $I$ a closed interval containing
$0$. We denote by $\Gamma(\L)$ the space of smooth sections of $\L$, $\Conn(\L)$ the space of metric connexions on $\L$, $\Omega^p$ the space of smooth $p$-forms on $\Sigma$, $*:\Omega^p\ra\Omega^{2-p}$ the Hodge isomorphism, $\d:\Omega^p\ra\Omega^{p+1}$ the exterior derivative and $\dstar=-*\d*$ its
$L^2$ adjoint. To a collection of curves $\phi:I\ra\Gamma(\L)$, $A:I\ra \Conn(\L)$, $A_0:I\ra\Omega^0$, $N:I\ra\Omega^0$ we associate the action
\beq\label{Sdef}
S=S_{\mbox{\tiny YMH}}+\kappa S_{\mbox{\tiny CS}}
\eeq
where
\bea
S_{\mbox{\tiny YMH}}&=&\frac12\int_I\bigg\{
\|D_0\phi\|^2-\|\d_A\phi\|^2+\|e\|^2-\|B\|^2+\|\dot{N}\|^2-\|\d N\|^2\nonumber \\ \label{SYMH}
&&\qquad\qquad\qquad-\frac14\left\|1+2\kappa N-|\phi|^2\right\|^2-\|N\phi\|^2\bigg\}dt
\eea
is a deformed Yang-Mills-Higgs action for which the Higgs field $\phi$ is coupled to the neutral scalar field $N$, 
and
\beq
S_{\mbox{\tiny CS}}=\frac12\int_I\bigg\{\ip{A_0,*B}+\ip{A,*e}\bigg\}dt
\eeq
is the Chern-Simons action.
Here $D_0=\cd_t-iA_0$, $\ip{\cdot,\cdot}$ denotes $L^2$ inner product on $(\Sigma,g)$, $\|\cdot\|$ the associated $L^2$ norm, $\kappa$ is a real parameter, and
$e:I\ra\Omega^1$, $B:I\ra\Omega^2$ are the electric and magnetic fields
\beq
e=\dot{A}-\d A_0,\qquad B=\d A.
\eeq
This is the Maxwell-Chern-Simons-Higgs action on spacetime $(I\times\Sigma,dt^2-g)$ \cite{leeleemin,hankim}. The Euler-Lagrange equations, satisifed by formal critical points of $S$, are
\bea
D_0^2\phi+\dstar_A\d_A\phi&=&\frac12(1+2\kappa N-|\phi|^2)\phi-N^2\phi,\label{EL1}\\
\ddot{N}+\Delta N&=&-\frac\kappa2(1+2\kappa N-|\phi|^2)-|\phi|^2N,\label{EL2}\\
\dot{e}+\dstar B-\kappa*e&=&h(i\phi,\d_A\phi)=:j\label{EL3}\\
-\delta e +\kappa *B&=&h(i\phi,D_0\phi)=:\rho\label{gauss}
\eea
where $j=h(i\phi,\d_A\phi)$ and $\rho=h(i\phi,D_0\phi)$ have the physical interpretation of electric current and charge density,
respectively. Here, and henceforth, $\Delta=\delta\d+\d\delta$ denotes the Hodge Laplacian, which has the opposite sign convention to that favoured by analysts (e.g.\ $\Delta=-(\cd_x^2+\cd_y^2)$ for functions on euclidean $\R^2$). The last of the field equations, 
(\ref{gauss}), obtained by varying $A_0$, is Gauss's law, and should be thought of
as a constraint on initial data which, if satisfied at $t=0$, automatically holds for all $t\in I$. Solutions of this system conserve the total energy
\bea
E&=&
\frac12\bigg\{\|D_0\phi\|^2+\|\d_A\phi\|^2+\|e\|^2
+\|B\|^2+\|\dot{N}\|^2+\|\d N\|^2\nonumber \\
&&\qquad\qquad\qquad
+\frac14\left\|1+2\kappa N-|\phi|^2\right\|^2+\|N\phi\|^2\bigg\}.\label{Edef}
\eea

There is a Bogomol'nyi type topological lower bound on the energy of any (possibly time-dependent) field configuration satisfying Gauss's law (\ref{gauss}). To see this, it is convenient to decompose $\d_A\phi$ into its $(1,0)$ and $(0,1)$ components (with respect to the almost complex structure on $\Sigma$ defined by $g$ and
its orientation),
\beq
\d_A\phi=:\cd_A\phi+\ol\cd_A\phi,
\eeq
and note the standard identity
\beq
*B|\phi|^2=|\cd_A\phi|^2-|\ol{\cd}_{A}\phi|^2.
\eeq
Then, for all fields satisfying (\ref{gauss}),
\bea
0&\leq&\frac12\|\dot{N}\|^2+\frac12\|e+\d N\|^2+\frac12\|*B-\frac12(1+2\kappa N-|\phi|^2)\|^2
+\frac12\|D_0\phi+iN\phi\|^2+\|\ol{\cd}_A\phi\|^2\nonumber\\
&=&E-\frac12\int_\Sigma B+\ip{N,\rho+\delta e-\kappa *B}\nonumber\\
&=&E-\pi n.
\eea
Hence
\beq
E\geq \pi n
\eeq
with equality if and only if
\bea
\dot{N}&=&0\\
e&=&-\d N\label{1}\\
D_0\phi&=&-iN\phi\label{2}\\
\ol{\cd}_{A}\phi&=&0\label{3}\\
*B&=&\frac12(1+2\kappa N-|\phi|^2).\label{4}
\eea
This system simplifies considerably if we make the gauge choice $A_0=N$. Then all fields are static (i.e.\ $t$-independent) and the system, together with (\ref{gauss}),
reduces to
\bea
\ol\cd_A\phi&=&0\label{bog1}\\
*B&=&\frac12(1+2\kappa N-|\phi|^2)\label{bog2}\\
\Delta N+\kappa *B +|\phi|^2 N&=&0\label{bog3}
\eea
a coupled system for $(\phi,A,N)$ which we call the Bogomol'nyi equations. 
It is straightforward to verify that any static solution of
(\ref{bog1})-(\ref{bog3}) satisfies the Euler-Lagrange equations (with $A_0=N$). It is known \cite{hankim} that the converse is false (that is, there exist static solutions of the Euler-Lagrange equations which do not satisfy the Bogomol'nyi equations). 

The rest of this paper concerns existence, regularity and uniqueness of solutions of the Bogomol'nyi equations (\ref{bog1})-(\ref{bog3}). A simple, but crucial, observation is that, for $\kappa=0$ and $\Vol(\Sigma)> 4\pi n$,
(\ref{bog2}) implies $\|\phi\|>0$, whence 
(\ref{bog3}) implies $N=0$, and (\ref{bog1}),(\ref{bog2}) reduce to the standard vortex equations for $\L$, 
\beq
\ol\cd_A\phi=0,\qquad *B=\frac12(1-|\phi|^2),\label{vor}
\eeq
studied in \cite{bra}. Hence, the only solutions of the model with $\kappa=0$ are the usual vortex solutions
augemented by $N=0$. This is unsurprising
given that, when $\kappa=0$ and $N=0$, our action (\ref{Sdef}) reduces to the usual Maxwell-Higgs action.

\section{Existence of vortices}\news

Let $(\phi,A,N)$ be a smooth solution of (\ref{bog1})-(\ref{bog3}). Then, by (\ref{bog1}), $\phi^{-1}(0)\subset\Sigma$ defines 
an effective divisor $D$ of degree $n$ (that is, an unordered collection of $n$ points on $\Sigma$, not necessarily distinct). Conversely,
let such a divisor $D$ be fixed. Then, provided $\Vol(\Sigma)>4\pi n$, the vortex equations (\ref{vor}) have a unique (up to gauge) solution with $\phi^{-1}(0)=D$, and this solution is smooth. Let us denote this solution $(\hat\phi,\hat{A})$. We have already observed that $(\hat\phi,\hat A,N=0)$ is trivially a solution of
(\ref{bog1})-(\ref{bog3}) for $\kappa=0$. Our aim is to prove, via the Implicit Function Theorem, that this trivial solution has a locally unique smooth continuation to $|\kappa|>0$, for all $|\kappa|$ sufficiently small.

\begin{thm} \label{ranhas} Let $D$ be an effective divisor of degree $n$ on $\Sigma$, and assume $\Vol(\Sigma)>4\pi n$. Then there exist $\kappa_*>0$ and a smooth curve
$(\phi,A,N):(-\kappa_*,\kappa_*)\ra \Gamma(\L)\times\Conn(\L)\times\Omega^0$, unique up to gauge, such that, for all $\kappa\in(-\kappa_*,\kappa_*)$,
$(\phi(\kappa),A(\kappa),N(\kappa))$ satisfies the Bogomol'nyi equations and $\phi(\kappa)^{-1}(0)=D$.
\end{thm}

\begin{proof} Let $(\hat\phi,\hat{A})$ be the solution of (\ref{vor}), unique up to gauge, with $\hat\phi^{-1}(0)=D$, whose existence was established in \cite{bra}. Every section $\phi$ with $\phi^{-1}(0)=D$ has a unique representative in its gauge orbit of the form $\phi=e^u\hat\phi$, where $u:\Sigma\ra\R$. This section satisfies (\ref{bog1}) if and only if $A=\hat{A}-*\d u$. Then $(\phi,A,N)$ satisfies (\ref{bog2}), (\ref{bog3}) if and only if
\bea
\Delta u+\frac{f}{2}(e^{2u}-1)-\kappa N&=&0\label{ariel1}\\
\Delta N+(\kappa^2+fe^{2u})N+\frac\kappa2(1-fe^{2u})&=&0\label{ariel2}
\eea
where $f:\Sigma\ra[0,\infty)$ is the smooth function $|\hat\phi|^2$ vanishing only on $D$. Clearly, $(\phi,A,N)=(e^u\hat\phi,\hat A-*\d u,N)$ is smooth if $(u,N)$ is. 
We have thus reduced the problem to proving the following Lemma which, for later
convenience, we formulate on an arbitrary compact Riemannian manifold.

\begin{lemma} \label{natgra}
Let $M$ be a compact Riemannian manifold of dimension $m\geq 2$, and $f$ a
smooth non-negative function on $M$ vanishing on a set of measure $0$. Then there exists $\kappa_*>0$ and a smooth
curve $(u,N):(-\kappa_*,\kappa_*)\ra C^\infty(M)\times C^\infty(M)$ with $(u,N)(0)=(0,0)$ such that, for all $\kappa\in
(-\kappa_*,\kappa_*)$,
$(u(\kappa),N(\kappa))$ satisfies (\ref{ariel1}), (\ref{ariel2}). Furthermore,
this curve is locally unique.
\end{lemma}

\begin{proof}
Let $H^k$ denote the space of real functions on $M$ whose derivatives up to order $k$ are
square integrable. This is a Hilbert space with respect to the inner product
\beq
\ip{u,v}_{H^k}=\ip{u,v}+\ip{\nabla u,\nabla v}+\cdots+\ip{\nabla^ku,\nabla^kv},
\eeq
where $\nabla$ denotes the Levi-Civita connexion \cite{aub}. Furthermore, for all $k\geq r$, where
\beq
r:=\left\lfloor\frac{m}{2}\right\rfloor+1,
\eeq
$H^k$ is a Banach algebra under pointwise multiplication: for all $u,v\in H^k$, $uv\in H^k$, and there
exists a constant $C(k)>0$, such that, for all $u,v\in H^k$, 
$\|uv\|_{H^k}\leq C(k)\|u\|_{H^k}\|v\|_{H^k}$. (This fact is well known, though we have been unable to find a proof of it in the literature. The analogous statement for $H^k(\Omega)$, where $\Omega$ is
a domain in $\R^m$ with the cone property, is proved in \cite[pp.\ 115--117]{ada}. This proof uses only Sobolev embeddings,
which hold equally well on a compact Riemannian manifold \cite[pp.\ 35, 44]{aub}, so works,
{\em mutatis mutandis}, in our setting also.) It follows from this that, for all $k\geq r$, the
exponential map $\exp:u\mapsto e^u$ is a smooth mapping $H^k\ra H^k$. One can deduce this directly, by observing that $\exp(u)$ is the limit of the absolutely convergent power series
$\sum_{n=0}^\infty u^n/n!$ in the Banach algebra $H^k$, or appeal to general results on smoothness of
composition maps (see \cite[p.\ 424]{donkro}, for example). 

It follows that $F:\R\oplus H^r\oplus H^r\ra H^{r-2}\oplus H^{r-2}$,
\beq\label{Fdef}
F(\kappa,u,N)=(\Delta u+\frac{f}{2}(e^{2u}-1)-\kappa N,\Delta N+(\kappa^2+fe^{2u})N+\frac\kappa2(1-fe^{2u})),
\eeq
is a smooth map between Banach spaces. To see this, one notes that $F$ is a composition of the maps $\Delta:H^r\ra H^{r-2}$,
multiplication by $f$
($H^r\ra H^r$, $u\mapsto fu$), projections (${\rm pr}_{1,2}:H^k\oplus H^{k'}\ra H^k,H^{k'}$), inclusions ($H^k\ra H^k\oplus H^{k'}$, $u\mapsto(u,0)$,   $H^k\ra H^{k'}\oplus H^{k}$, $u\mapsto(0,u)$, $H^k\hra H^{k-1}$), $\exp:H^r\ra H^r$, and the multiplication map ($\mu:H^r\oplus H^r\ra H^r$,
$\mu:\R\oplus H^r\ra H^r$). All but the last two of these are bounded linear, hence trivially smooth, and we have already observed that $\exp$ is smooth. Clearly $\mu$, being bilinear, is continuous. Its differential
$\d\mu_{(u,v)}(u',v')=\mu(u,u')+\mu(v,v')$, is a bounded linear map (for each fixed $(u,v)$) which is continuous, by continuity of $\mu$, so $\mu$ is $C^1$, and hence smooth by the chain rule. 

By construction, our PDE system (\ref{ariel1}),(\ref{ariel2}) is equivalent to
$F(\kappa,u,N)=(0,0)$. Let $F|:H^r\oplus H^r\ra H^{r-2}\oplus H^{r-2}$, $F|(u,N)=F(0,u,N)$.
Now $F(0,0,0)=(0,0)$, and $dF|_{(0,0)}:H^r\oplus H^r\ra H^{r-2}\oplus H^{r-2}$ is $L\oplus L$ where $L:H^r\ra H^{r-2}$ is the second order linear elliptic operator
\beq
L=\Delta+f.
\eeq
Since $f\geq0$ and vanishes only on a measure zero set, if $Lu=0$ then
\beq
0=\ip{u,\Delta u}+\ip{u,fu}\geq \ip{u,fu}=\|\sqrt{f}u\|^2\geq 0
\eeq
so $u=0$ almost everywhere. Hence $\ker L=\{0\}$, so, by the standard elliptic estimate for $L$ 
\cite[p.\ 423]{donkro}, there exists $C>0$ such that, for all $u\in H^r$,
\beq
\|u\|_{H^r}\leq C\|Lu\|_{H^{r-2}}.
\eeq
It follows \cite[p.\ 42]{hilphi} that $L$ is invertible, that is, $L$ is bijective and its inverse $L^{-1}:H^{r-2}\ra H^r$ is a bounded linear map. Hence, $\d F|_{(0,0)}$ is likewise invertible, with inverse $L^{-1}\oplus L^{-1}$. Thus, we may apply the Implicit Function Theorem \cite[p.\ 420]{donkro}
to $F$ at the point $(0,0,0)$: there exist $\kappa_*>0$ and a 
unique smooth map $x:(-\kappa_*,\kappa_*)\ra H^r\oplus H^r$ such that $x(0)=(0,0)$ and, for all $\kappa\in (-\kappa_*,\kappa_*)$, $F(\kappa,x(\kappa))=0$. This establishes existence of a unique smooth curve of solutions $x(\kappa)=(u(\kappa),N(\kappa))$ of (\ref{ariel1}),(\ref{ariel2}) in $H^r\times H^r$ with $u(0)=N(0)=0$.

Up to this point, we know that the curve of solutions $(u(\kappa),N(\kappa))$ is smooth with respect to $\kappa$, but we do not yet know that, for fixed $\kappa$, each of $u(\kappa)$, $N(\kappa)$ is a smooth function on $M$. This follows from a standard bootstrap argument. Dropping the argument $\kappa$ from $u(\kappa)$, $N(\kappa)$, we know that they lie in $H^r$ and satisfy (\ref{ariel1}),
(\ref{ariel2}).  Hence, by the standard elliptic estimate for $\Delta$ \cite[p.\ 423]{donkro}
 if $u,N\in H^{k\geq r}$, then
\beq
\|u\|_{H^{k+1}}\leq C(\|\Delta u\|_{H^{k-1}}+\|u\|)\leq C(\|e^{2u}\|_{H^{k-1}}+1+|\kappa|\|N\|_{H^k}+\|u\|_{H^k})<\infty,
\eeq
using (\ref{ariel1}), so $u\in H^{k+1}$,
and
\beq
\|N\|_{H^{k+1}}\leq C(\|\Delta N\|_{H^{k-1}}+\|N\|)\leq C(\kappa^2\|N\|_{H^k}+\|e^{2u}N\|_{H^{k}}+|\kappa|+|\kappa|\|e^{2u}\|_{H^{k}})<\infty
\eeq
using (\ref{ariel2}) and the algebra property of $H^k$. Hence, by induction on $k$, $(u,N)\in H^k\oplus H^k$ for all $k\geq r$. But, for all $k\geq r$ we have the continuous Sobolev embedding
$H^k\hra C^{k-r}$
\cite[pp.\ 35, 44]{aub}:
every $u\in H^k$ is $(k-r)$ times continuously differentiable, and there
exists $C(k)>0$ such that, for all $u\in H^k$,
\beq\label{sob}
\|u\|_{C^{k-r}}=\sup_{x\in M}\{|u(x)|, |\nabla u(x)|,\ldots,|\nabla^ku(x)\}\leq C(k)\|f\|_{H^k}.
\eeq
So $U,N\in C^{k-r}$ for all $k\geq r$, and hence are smooth.
\end{proof}

This completes the proof of Theorem \ref{ranhas}.
\end{proof} 

\section{Bounds on $\kappa_*$}\news

In this section we find upper and lower bounds on $\kappa_*$, the maximal Chern-Simons coupling to which we can smoothly continue a standard vortex. In principle, $\kappa_*$ may depend on $D$, the choice of vortex positions. We shall show that, for fixed $n$, there are global upper and lower bounds on $\kappa_*(D)$, independent of $D$. Much the simpler is the upper bound.

\begin{thm}\label{thm2} Assume that the Bogomol'nyi equations (\ref{bog1}), (\ref{bog2}), (\ref{bog3}) have a smooth solution. Then 
$$
\kappa^2\leq \frac{\Vol(\Sigma)}{4\pi n}-1.
$$
\end{thm}

\begin{proof} It is clear that $(\kappa,\phi,A,N)$ satisfies the Bogomol'nyi equations if and only if $(-\kappa,\phi,A,-N)$ does, so we may assume, without loss of generality, that
$\kappa\geq 0$. Integrating (\ref{bog3}) over $\Sigma$, we see that
\beq\label{ft2.5}
\ip{N,|\phi|^2}=-\kappa\int_\Sigma B=-2\pi\kappa n.
\eeq
The $L^2$ inner product of (\ref{bog3}) with $N$ implies
$\|\d N\|^2+\|N\phi\|^2+\kappa\ip{*B,N}=0$, whence
\beq\label{ft4}
\ip{*B,N}\leq0.
\eeq
Similarly, the $L^2$ inner product of (\ref{bog2}) with $N$ implies
\beq\label{ft3}
\ip{*B,N}=\frac12\int_\Sigma{N}-\frac12\ip{N,|\phi|^2}+\kappa\|N\|^2,
\eeq
and the integral of (\ref{bog2}) over $\Sigma$ yields
\beq\label{ft5}
\kappa\int_\Sigma N\geq 2\pi n -\frac12\Vol(\Sigma).
\eeq
Hence, by (\ref{ft4}),
\bea
0\geq\kappa\ip{*B,N}&\geq&\frac\kappa2\int_\Sigma N-\frac\kappa2\ip{N,|\phi|^2}\qquad\mbox{by (\ref{ft3})} \nonumber \\
&=&\frac\kappa2\int_\Sigma N+\pi\kappa^2n\qquad\mbox{by (\ref{ft2.5})} \nonumber \\
&\geq&\pi n - \frac14\Vol(\Sigma)+\pi\kappa^2 n\qquad\mbox{by (\ref{ft5})}\nonumber
\eea
and the claim immediately follows.
\end{proof}

\begin{remark}
In the case where $\Sigma$ is a flat torus, Ricciardi and Tarantello \cite{rictar}
prove a considerably stronger upper bound on $\kappa$, equivalent, in our conventions, to
\beq\label{rictarbounds}
\Vol(\Sigma)\geq 4\pi n\qquad\mbox{and}\qquad
\kappa^2 \leq \frac{\Vol(\Sigma)}{16\pi n}\left(1-\frac{4\pi n}{\Vol(\Sigma)}\right)^2.
\eeq
There is no reason to suppose that their proof cannot be adapted to deal with the case of general $\Sigma$, so one expects this stronger bound to be true more 
generally. The proof, however, requires pointwise control of $|\phi|$ and $N$ and so is, inevitably, much longer and more difficult than the proof of Theorem
\ref{thm2}. Since neither bound is likely to be sharp, we shall not attempt to generalize their argument here. It is the {\em existence} of an upper bound on $\kappa_*$ which is of primary interest to us.
\end{remark}

We turn now to the existence of a positive {\em lower} bound on $\kappa_*$, independent of $D$. 
Recall that the set $B(V,W)$ of bounded linear maps between Banach spaces $V,W$ is itself a Banach space with respect to the norm $\|A\|_{V\ra W}=
\sup\{\|Au\|_W\: :\: \|u\|_V=1\}$. We will omit the subscript $V\ra W$ where no confusion is possible. Let
\beq\label{csobdef}
\csob=\|\iota\|_{H^2\ra L^2},
\eeq the norm of the inclusion
$\iota:H^2\ra C^0$, or, equivalently, the optimal (smallest) constant in (\ref{sob}) for $k=m=2$. The space $\M_n$ of effective divisors of degree $n$ on $\Sigma$ is a compact topological space homeomorphic to $\sym{n}\Sigma$, the $n$-fold symmetric product of $\Sigma$. Assuming $4\pi n< \Vol(\Sigma)$, let $f_D\in\Omega^0$ denote, for
each $D\in \M_n$,  the squared
length of the Higgs field $\hat\phi$ of the unique (up to gauge) standard vortex vanishing on $D$. The map $\M_n\ra C^0$, $D\mapsto f_D$ is continuous \cite{garpra},
from which it immediately follows that
\beq
\Ll:\M_n\ra B(H^2,L^2),\qquad D\mapsto L_D=\Delta+f_D
\eeq
is continuous. 
Note that $L_D$ is the linear map we previously called $L$, introduced in the proof of Lemma \ref{natgra}. The new notation calls attention to the fact that this operator depends on the divisor $D$. As previously argued, each operator 
$L_D$ is 
invertible, and inversion is a continuous map \cite[p.\ 170]{bol} from the subset of invertible operators in $B(V,W)$ to $B(W,V)$, so it follows that
\beq
\Ll':\M_n\ra B(L^2,H^2),\qquad D\mapsto L_D^{-1}
\eeq
is also continuous. Since $\M_n$ is compact, $\Ll'(\M_n)$ is bounded, that is
\beq\label{C*def}
C_*=\sup_{D\in\M_n}\|L_D^{-1}\|<\infty.
\eeq
Note that $C_*>0$ depends only on $(\Sigma,g)$ and $n$. It would be useful to have a more explicit upper bound on $C_*$ but we have been unable to find one. By contrast, 
a simple bound on $\Ll(\M_n)$ is easily obtained. 
For all $D$,
\beq
\|L_D\|\leq \|\Delta\|+\sup\{\|f_D u\|\: :\: \|u\|_{H^2}=1\}\leq
\sqrt{2}+\csob\|f_D\|\leq \sqrt{2}+\csob\cvol,
\eeq
where
\beq
\cvol:=\sqrt{\Vol(\Sigma)-4\pi n},
\eeq
since, by (\ref{vor}),
\beq\label{fest}
\|f_D\|^2=\int_\Sigma(1-2*B)^2=\Vol(\Sigma)-4\int_\Sigma B+4\|B\|^2\leq\Vol(\Sigma)-8\pi n+4E=\cvol^2.
\eeq

Having introduced the constants $\csob$ (depending only on $(\Sigma,g)$) and $C_*,\cvol$ (depending also on $n$), we can establish the existence of a global (i.e.\ independent of $D$) lower bound on $\kappa_*$. For fixed $D\in \M_n$, let $F:\R\oplus H^2\oplus H^2\ra L^2\oplus L^2$ be the smooth mapping between Banach spaces defined in equation (\ref{Fdef}), and for fixed $\kappa\in\R$, let $F_\kappa|$ denote the map $(u,N)\mapsto F(\kappa,u,N)$. The first step is to show that $\d F_\kappa|$ is uniformly invertible on a small ball, of radius independent of $D$.

\begin{lemma}\label{fma} Let $D\in\M_n$ and $\Vol(\Sigma)>4\pi n$. Then there exists a constant $\eps\in(0,1]$, independent of $D$, such that, for all $\kappa\in(-\eps,\eps)$,
and $(u,N)\in B_\eps(0)\subset H^2\oplus H^2$, the linear map $\d F_\kappa|_{(u,N)}:H^2\oplus H^2\ra L^2\oplus L^2$ is invertible, and
$$
\|(\d F_\kappa|_{(u,N)})^{-1}\|\leq 2C_*.
$$
\end{lemma}

\begin{proof} Choose $\kappa\in[-1,1]$ and $u,N\in H^2$ with $\|u\|_{H^2}\leq 1$, $\|N\|_{H^2}\leq 1$. Then $\d F_{\kappa}|=L_D\oplus L_D+T$, where
\beq
T(u',N')=(f_D(e^{2u}-1)u'-\kappa N',f_D(e^{2u}-1)N'+\kappa^2N'+2f_De^{2u}Nu'-\kappa f_De^{2u}u').
\eeq
Now
\bea
\|T(u',N')\|_{L^2\oplus L^2}&\leq&\|e^{2u}-1\|_{C^0}(\|u'\|_{C^0}+\|N'\|_{C^0})\|f_D\|+2|\kappa|\|N'\|+\nonumber \\
&&2\|e^{2u}\|_{C^0}\|N\|_{C^0}\|u'\|_{C^0}\|f_D\|+|\kappa|\|e^{2u}\|_{C^0}
\|u'\|_{C^0}\|f_D\|.
\eea
Since $\|u\|_{H^2}\leq 1$, $-\csob\leq u(p)\leq \csob$ for all $p\in\Sigma$, whence
\beq
\|e^{2u}-1\|_{C^0}\leq 2e^{2\csob}\|u\|_{C^0}\leq C\|u\|_{H^2},
\eeq
where $C=2\csob e^{2\csob}$.
It follows that
\bea
\|T\|&=&\sup\{\|T(u',N')\|\: :\: \|u'\|_{H^2}^2+\|N'\|_{H^2}^2=1\}\nonumber \\
&\leq&
C\|u\|_{H^2}(\csob+\csob)\|f_D\|+2|\kappa|+C\|N\|_{H^2}\csob\|f_D\|+C|\kappa|\|f_D\|\nonumber\\
&\leq& 2 C\csob\cvol(\|u\|_{H^2}+\|N\|_{H^2})+(2+C\cvol)|\kappa|,\label{Test}
\eea
by (\ref{fest}). Let 
\beq\label{epsdef}
\eps=\min\left\{1, \frac{C_*^{-1}}{16C\csob\cvol},\frac{C_*^{-1}}{8+4C\cvol}\right\}>0,
\eeq
and note that this is independent of $D$. Then, for all $(\kappa,(u,N))\in (-\eps,\eps)\times B_{\eps}(0)$, since $\eps\leq 1$ the estimate above, (\ref{Test}),
 for $\|T\|$ holds, and so
\beq
\|T\|<\frac1{2C_*}\leq\frac12\frac{1}{\|L_D^{-1}\|}=\frac12\frac{1}{\|(L_D\oplus L_D)^{-1}\|}.
\eeq
It follows that $\|(L_D\oplus L_D)^{-1}T)\|<\frac12$, and
hence \cite[p.\ 169]{bol} that
\beq
\d F_{\kappa}|=(L_D\oplus L_D)(1+(L_D\oplus L_D)^{-1}T)
\eeq
is invertible with inverse
\beq
(\d F_{\kappa}|)^{-1}=\left(\sum_{k=0}^\infty (-1)^k((L_D\oplus L_D)^{-1}T)^k\right)(L_D\oplus L_D)^{-1},
\eeq
whose norm satisfies
\bea
\|(\d F_{\kappa}|)^{-1}\|&\leq&\left(\sum_{k=0}^\infty \frac{1}{2^k}\right)\|(L_D\oplus L_D)^{-1}\|=2\|L_D^{-1}\|\leq 2C_*.
\eea
\end{proof}

\begin{thm} Assume $\Vol(\Sigma)>4\pi n$.
For each $D\in\M_n$, let $(-\kappa_*(D),\kappa_*(D))>0$ be the maximal open interval on which a smooth curve of solutions of the Bogomol'nyi equations with
$\phi^{-1}(0)=D$ exists. Then
$$
\kappa_*(D)\geq\eps\min\left\{1,\frac{1}{C_*(6+(1+e^{2\csob})\sqrt{\Vol(\Sigma)})}\right\}
$$
where $\csob,C_*,\eps$ are the positive constants, independent of $D$, defined in (\ref{csobdef}),(\ref{C*def}),(\ref{epsdef}).
\end{thm}

\begin{proof}
Denote by $x(\kappa)=(u(\kappa),N(\kappa))$ the solution curve in $H^2\oplus H^2$ whose existence was established in Theorem \ref{ranhas}. This curve starts at $x(0)=0$ and exists whilever $\d F_\kappa|_{x(\kappa)}$
is invertible. By Lemma \ref{fma}, $\d F_\kappa|_{x(\kappa)}$ is invertible if $|\kappa|<\eps$ and $\|x(\kappa)\|_{H^2\oplus H^2}<\eps$. Hence, either $\kappa_*(D)\geq
\eps$, or the curve exits the ball $B_\eps(0)\subset H^2\oplus H^2$ at some ``time" $\kappa'\leq\kappa_*(D)<\eps$. In the latter case, its ``speed", must, at some time prior
to its first exit,
be at least $\eps/\kappa_*(D)$. That is, there exists $0<\kappa_1<\kappa_*(D)<\eps$ such that $x(\kappa_1)\in B_\eps(0)$ and
\beq
\|\dot{x}(\kappa_1)\|_{H^2\oplus H^2}\geq\frac\eps{\kappa_*(D)}.
\eeq
Now $F(\kappa,x(\kappa))=0$ for all $\kappa$, so $x(\kappa)$ satisfies the ODE
\beq
\dot{x}(\kappa)=-(dF_{\kappa}|_{x(\kappa)})^{-1}\frac{\cd F}{\cd\kappa}\bigg|_{\kappa,x(\kappa)}.
\eeq
Hence
\beq
\|\dot{x}(\kappa_1)\|\ignore{\leq \|(dF_{\kappa_1}|_{x(\kappa_1)})^{-1}\|\left\|\frac{\cd F}{\cd\kappa}\bigg|_{\kappa_1,x(\kappa_1)}\right\|_{L^2\oplus L^2}}
\leq 2C_*\left\|\frac{\cd F}{\cd\kappa}\bigg|_{\kappa_1,x(\kappa_1)}\right\|_{L^2\oplus L^2}
\eeq
by Lemma \ref{fma}. Now
\beq
\frac{\cd F}{\cd\kappa}=(-N,2\kappa N+\frac12(1-f_De^{2u}))
\eeq
so
\bea
\left\|\frac{\cd F}{\cd\kappa}\bigg|_{\kappa_1,x(\kappa_1)}\right\|_{L^2\oplus L^2}&\leq& (1+2\kappa_1)\eps+\frac12(\|1\|+\|e^{2u}\|_{C^0}\|f_D\|)\nonumber \\
&\leq& 3+\frac12(1+e^{2\csob})\sqrt{\Vol(\Sigma)}.
\eea
Hence, either $\kappa_*(D)\geq\eps$, or
\bea
\frac\eps{\kappa_*(D)}\leq 2C_*(3+\frac12(1+e^{2\csob})\sqrt{\Vol(\Sigma)}),
\eea
which establishes the claim.
\end{proof}

\section{Vortices on $S^2$: a numerical study}
\news

In order to study the dependence of $\kappa_*(D)$ on the divisor $D$, we consider the case where $\Sigma$ is the round sphere of radius $R > \sqrt{n}$ and use numerical techniques to investigate the deformed solutions away from  $\kappa=0$. Now $(\L,h)$ is a degree $n\geq 1$ hermitian line bundle over $S^2$. Let $U_\pm=S^2\less\{(0,0,\mp1)\}$
and $\eps_\pm$ be unit length sections of $\L$ on $U_\pm$ such that $\eps_-=e^{in\theta}\eps_+$ on $U_+\cap U_-$, where $\theta$ is the azimuthal angle around the North-South axis on $S^2$. Then any section $\phi$ of $\L$ is determined by a pair of functions $\phi_\pm:U_\pm\ra\C$ satisfying 
\beq\label{ranhas1}
\phi_+=e^{in\theta}\phi_-
\eeq
on $U_+\cap U_-$. A unitary connexion on $\L$ is represented by a pair of real one forms $A_\pm\in\Omega^1(U_\pm)$ satisfying
\beq\label{ranhas2}
A_+=A_-+n\d\theta
\eeq
 on $U_+\cap U_-$.  On $U_\pm$ we use the stereographic coordinate obtained by projection from $(0,0,\mp 1)$, denoted $z_\pm=r_\pm e^{i\theta_\pm}$. Of course $\theta=\theta_+$, and 
$z_-=1/z_+$ on $U_+\cap U_-$, so $\theta_-=-\theta_+$.

Consider the case where $D=n_+(0,0,1)+n_-(0,0,-1)$, that is, the North pole with multiplicity $n_+\geq 0$ and the 
South pole with multiplicity $n_-\geq 0$, where $n_++n_-=n$. Then, by the rotational equivariance of the system,  we may choose gauge so that
\beq
\phi_\pm=f_\pm(r_\pm)e^{in_\pm\theta_\pm},\quad A_\pm=a_\pm(r_\pm)\d\theta_\pm.
\eeq
We also use $N_\pm$ to denote the restriction of $N$ to $U_\pm$, considered as a function of $r_\pm$. The
Bogomol'nyi equations (\ref{bog1}), (\ref{bog2}), (\ref{bog3}) then reduce to a system of ODEs, namely,
\begin{align}
\frac{df_\pm}{dr} &= \frac{1}{r}\left(n_\pm-a_\pm(r)\right)f_\pm(r)\nonumber \\
\frac{da_\pm}{dr} &= \frac{r}{2}\Omega(r)\left(1+2\kappa N_\pm(r) - f_\pm(r)^2\right)\nonumber \\
\frac{d^2N_\pm}{dr^2} &= \Omega(r)\left(N_\pm(r)(\kappa^2 + f_\pm(r)^2) + \frac{\kappa}{2}\left(1-f_\pm(r)^2\right)\right) -\frac{1}{r}\frac{dN_\pm}{dr},\label{ranhas3}
\end{align}
where $\Omega(r) = \frac{4R^2}{(1+r^2)^2}$ is the conformal factor on $S^2$. We must solve these six ODEs on the interval $[0,1]$ subject to matching conditions at $r=1$ (corresponding, in both coordinate patches, to the equator of $S^2$) determined by (\ref{ranhas1}), (\ref{ranhas2}). These are
\beq\label{ranhas4}
f_-(1)=f_+(1),\quad a_-(1)-n_-=n_+-a_+(1),\quad N_-(1)=N_+(1),\quad N_-'(1)=-N_+'(1).
\eeq
Owing to the coordinate singularities of the system (\ref{ranhas3}) at $r=0$, we must step away to $r=\delta>0$, small, by constructing a power series solution of (\ref{ranhas3}) about $0$. One finds that
\bea
f_\pm(r)&=& q_\pm r^{n_\pm}+\cdots,\\
a_\pm(r)&=& \left\lbrace\begin{array}{lr}
(1+2\kappa p_\pm-q_\pm^2)R^2r^2 +\cdots & \text{if $n_\pm$ = 0}\\
(1+2\kappa p_\pm)R^2r^2 +\cdots & \text{otherwise}
\end{array}\right.\\
N_\pm(r)&=& \left\lbrace\begin{array}{lr}
p_\pm + \frac{1}{2}(\kappa + 2\kappa^2p_\pm - \kappa q_\pm^2 + 2p_\pm q_\pm^2)R^2r^2 +\cdots & \text{if $n_\pm$ = 0}\\
p_\pm + \frac{1}{2}(\kappa + 2\kappa^2p_\pm)R^2r^2 +\cdots & \text{otherwise}
\end{array}\right.
\eea
so solutions on $[\delta,1]$ are uniquely determined by four shooting parameters $Z=(q_+,q_-,p_+,p_-)$. Finding a global solution on $S^2$ then corresponds to finding a zero of the map $M:\R^4\ra\R^4$,
\beq
M:Z\mapsto(f_+(1)-f_-(1),a_+(1)+a_-(1)-n,N_+(1)-N_-(1),N_+'(1)+N_-'(1)).
\eeq
We use  a fourth order Runge-Kutta method to solve (\ref{ranhas3}), and hence evaluate the map $M$, for a given
$Z$, and a Newton-Raphson scheme to solve $M(Z)=0$ in the initial case $\kappa=0$. In this way, we can numerically construct undeformed solutions with $n_+$ vortices at the North pole, and $n_-$ at the South pole. 

We now allow $\kappa$ to vary, so that the shooting map becomes a function $\R^5\ra\R^4$,
$(\kappa,Z)\mapsto M(\kappa,Z)$. We seek to construct a smooth curve in $\R^5$, $(\kappa(s),Z(s))$ satisfying
$M(\kappa(s),Z(s))=0$ with $\kappa(0)=0$ and $Z(0)$ the shooting parameters of the undeformed vortex. To do this,
we use {\em pseudo-arclength continuation}. Having obtained one point on the curve $(\kappa_0,Z_0)$, we
construct a tangent vector $(\dot\kappa_0,\dot Z_0)$ to the curve at $(\kappa_0,Z_0)$ and then seek a nearby point $(\kappa,Z)$ on the curve satisfying
\begin{align*}
\dot{Z}_0\cdot \left(Z-Z_0\right) + \dot{\kappa}_0\left(\kappa-\kappa_0\right) = \delta s,
\end{align*}
where $\delta s>0$ is a small constant. This is an approximation of the arc-length condition
\begin{align*}
||\dot{Z}||^2 + |\dot{\kappa}|^2 = 1.
\end{align*}
We achieve this by seeking a zero of the augmented map $\wt{M}:\R^5\ra\R^5$,
\beq
\wt{M}(\kappa,Z)=(M(\kappa,Z),\dot{Z}_0\cdot \left(Z-Z_0\right) + \dot{\kappa}_0\left(\kappa-\kappa_0\right) - \delta s)
\eeq
again using a Newton-Raphson scheme (and a Runge-Kutta method to evaluate $M$).

For each equivariant divisor and sphere radius $R$, we used the known solution at $\kappa=0$ and the initial unit tangent vector in the positive $\kappa$ direction in order to begin the continuation process. For the subsequent points, we use the difference between the solution $(\kappa_n,Z_n)$ and the previous solution 
$(\kappa_{n-1},Z_{n-1})$ to approximate the tangent vector at $(\kappa_n,{Z_n})$. 
In every case we find that the solution curves exhibit a single global maximum value of $\kappa(s)$, that is,
$\kappa(s)$ initially increases monotonically, before reaching a turning point, after which it decreases 
monotonically towards $0$.  The numerics suggest that, after its turning point, $\kappa(s)$ approaches $0$ arbitrarily closely but never reaches $0$. This is consistent both with the uniqueness of undeformed vortices
\cite{bra} and with the existence, for all small $\kappa>0$ of (at least) {\em two} distinct vortex solutions for each divisor, as suggested by Ricciardi and Tarantello's analysis of the model on $T^2$ \cite{rictar}. We see then, that on compact $\Sigma$, these two distinct solutions probably lie on the same solution branch, along which
$\kappa$ has a turning point. The solution curves for $(n_+,n_-)=(2,0)$ and $(n_+,n_-)=(1,1)$ with $R=4$ are depicted in Figure \ref{fig:curve}. The behaviour of the solutions as $s$ approaches its supremum $s_1$ (and $\kappa$ approaches $0$) is interesting: the vortex appears to degenerate into a configuration with vanishing Higgs and electric fields and
spatially constant magnetic field. The value of the constant magnetic field's spatially constant limit is $\frac{n}{2R^2}$, since the magnetic flux is still quantised. Furthermore, the neutral scalar field becomes spatially constant and diverges to $-\infty$, in such a way that $\kappa N\ra \frac{n}{2R^2}-\frac12$, in this instance to $-7/16$, as is shown by Figure \ref{fig:kappaN}. Note that in the large $R$ limit, $(\phi,N)$ approaches the
alternative vacuum $\phi=0$, $N=-(2\kappa)^{-1}$, which diverges as $\kappa\ra 0$. This suggests a
possible link between the extra vortices found on compact domains and the nontopological
``vortices" found on the plane (which tend, as $r\ra\infty$ to this alternative vacuum)
\cite{leeleemin}. The limit where $\kappa(s)$ returns to $0$, while reminiscent of dissolving vortices in the Bradlow limit \cite{bapman}, is quite subtle, since the Bogomol'nyi equations do not support
solutions of this type.

\begin{figure}
\begin{subfigure}{0.5\textwidth}
  \centering
  \begin{overpic}[width=0.8\textwidth]{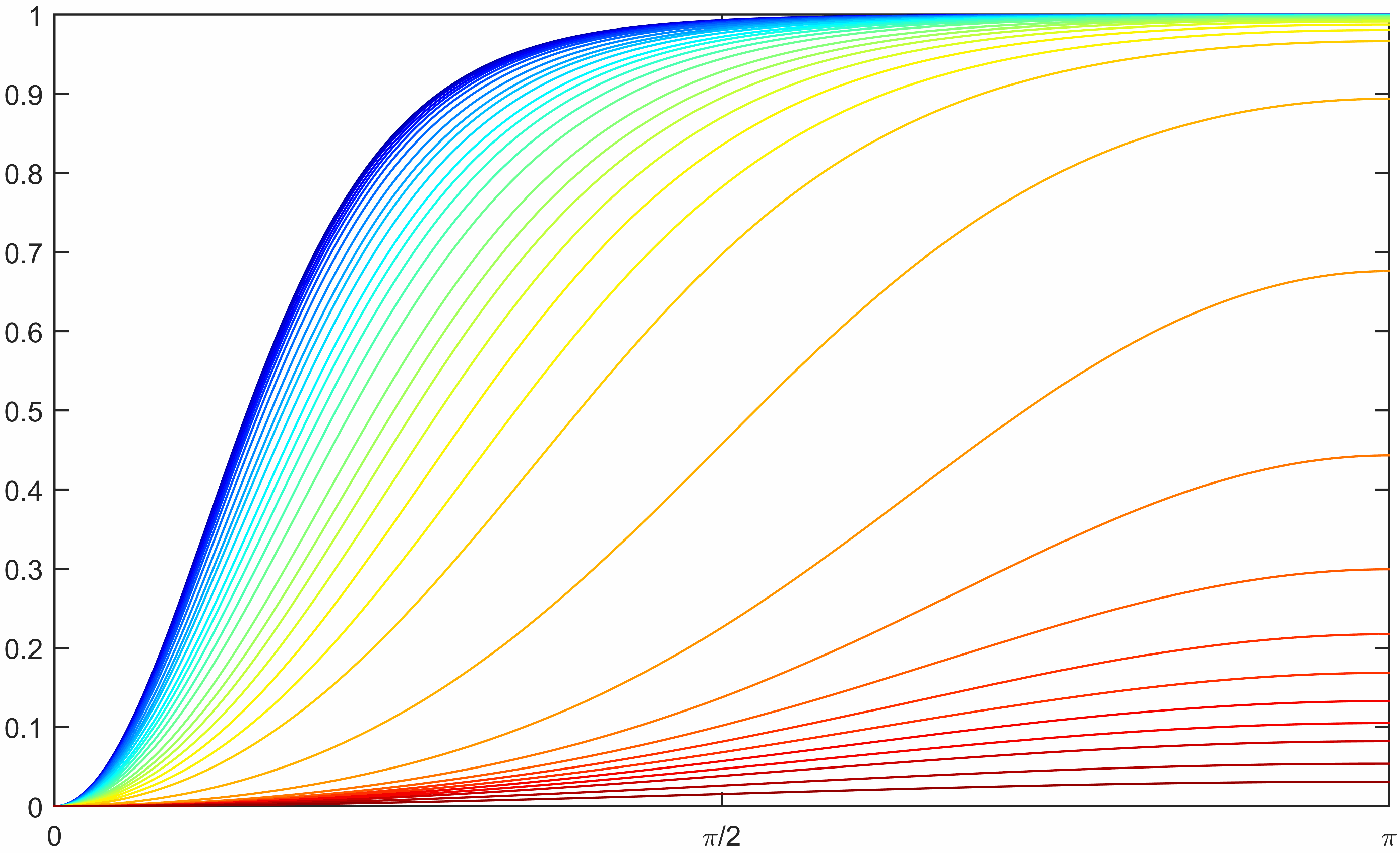}
     	\put(50,-8){$\Theta$}
       	\put(-8,32){\begin{turn}{90}$f$\end{turn}}
  \end{overpic}
        \caption{\phantom{Enough empty space to clear room}}
        \label{fig:1a}
  \par\medskip 
  \begin{overpic}[width=0.8\textwidth]{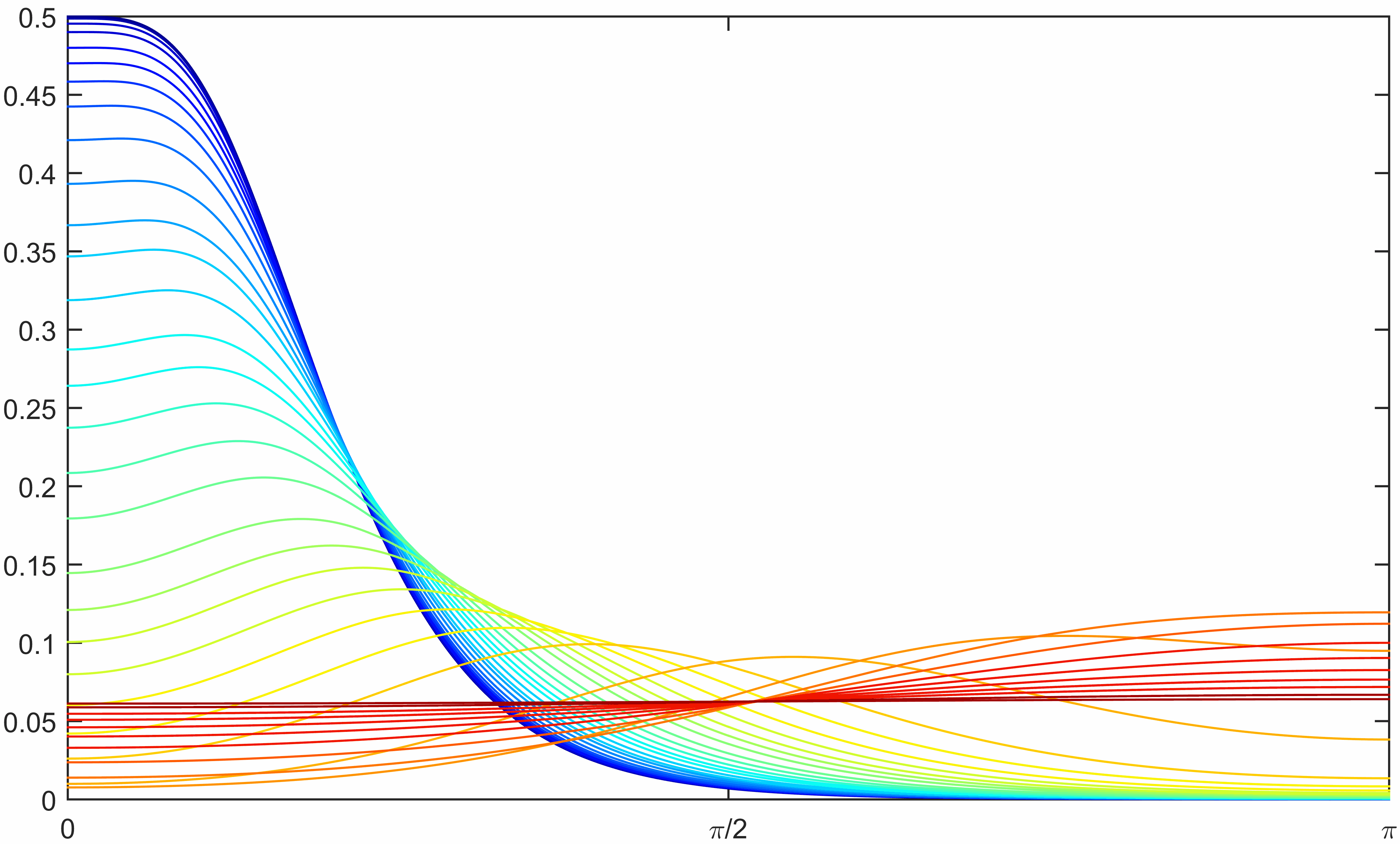}
     	\put(50,-8){$\Theta$}
       	\put(-8,28){\begin{turn}{90}$\star B$\end{turn}}
  \end{overpic}
        \caption{\phantom{Enough empty space to clear room}}
        \label{fig:1b}
  \par\medskip 
  \begin{overpic}[width=0.8\textwidth]{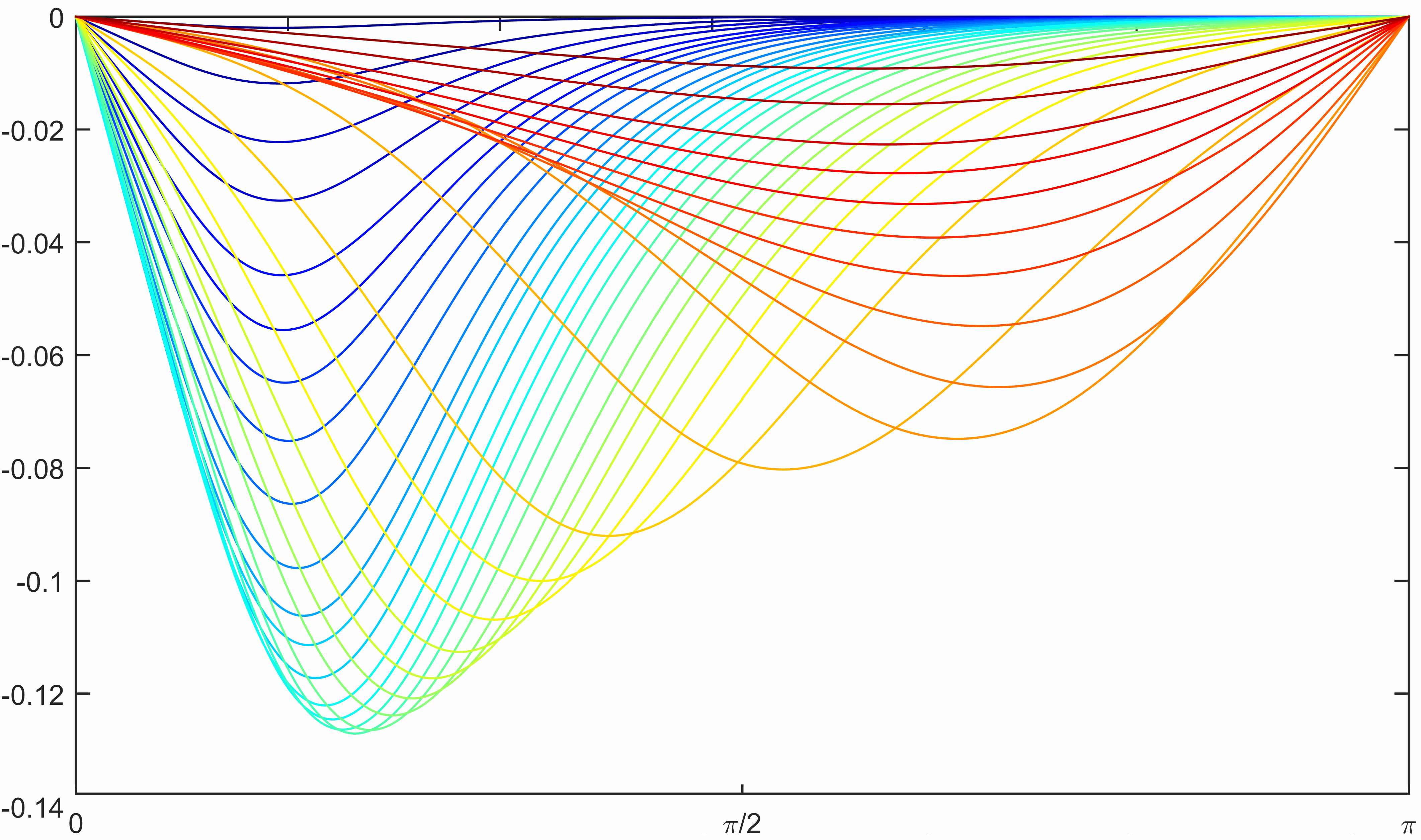}
     	\put(50,-8){$\Theta$}
       	\put(-8,18){\begin{turn}{90}$\left\langle e, R\d\Theta\right\rangle$\end{turn}}
  \end{overpic}
        \caption{\phantom{Enough empty space to clear room}}
        \label{fig:1c}
  \par\medskip 
  \begin{overpic}[width=0.8\textwidth]{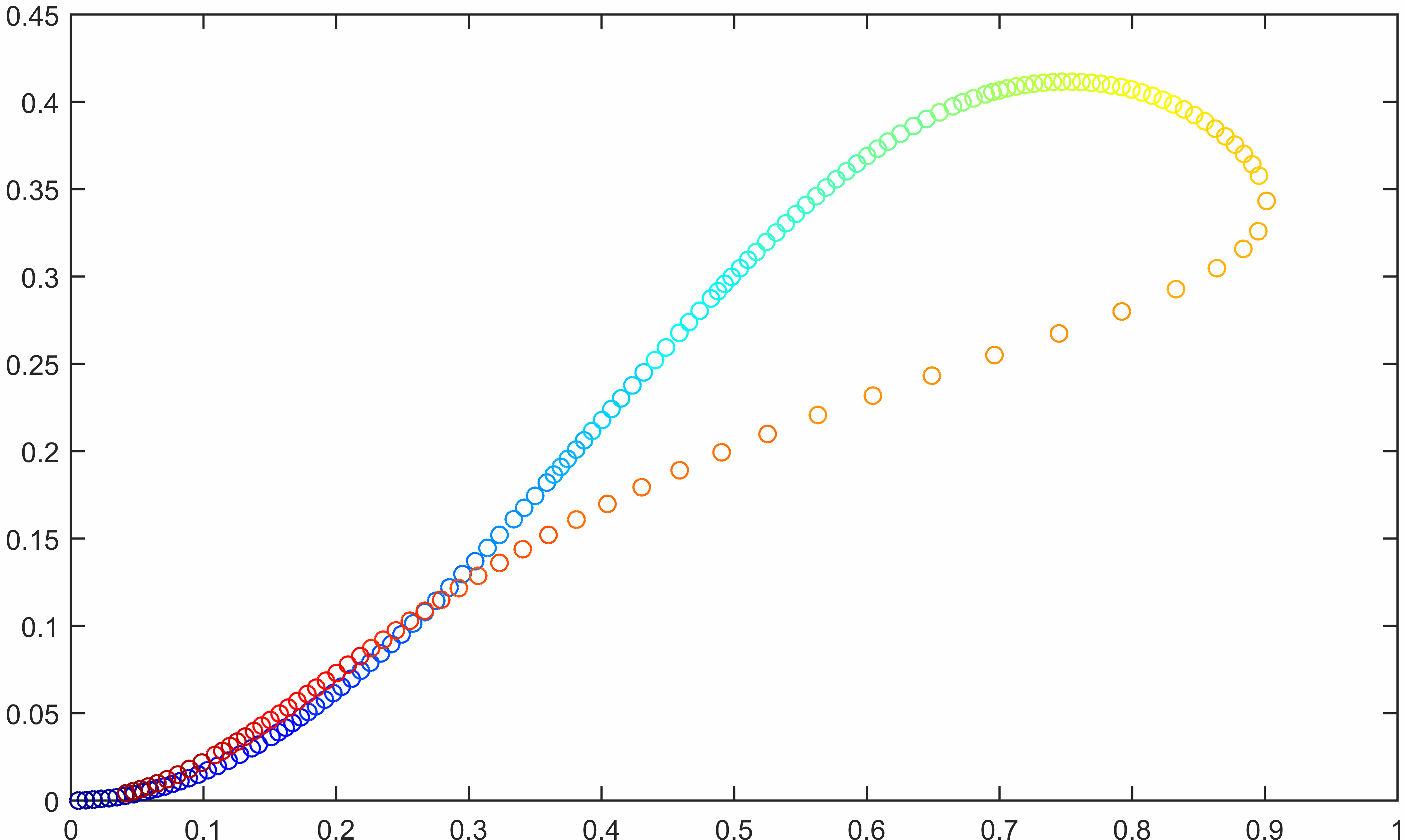}
     	\put(50,-8){$\kappa$}
       	\put(-8,24){\begin{turn}{90}$E_{elec}$\end{turn}}
  \end{overpic}
        \caption{\phantom{Enough empty space to clear room}}
        \label{fig:1d}
\end{subfigure}
\hspace*{\fill}
\begin{subfigure}{0.5\textwidth}
\centering
  \begin{overpic}[width=0.8\textwidth]{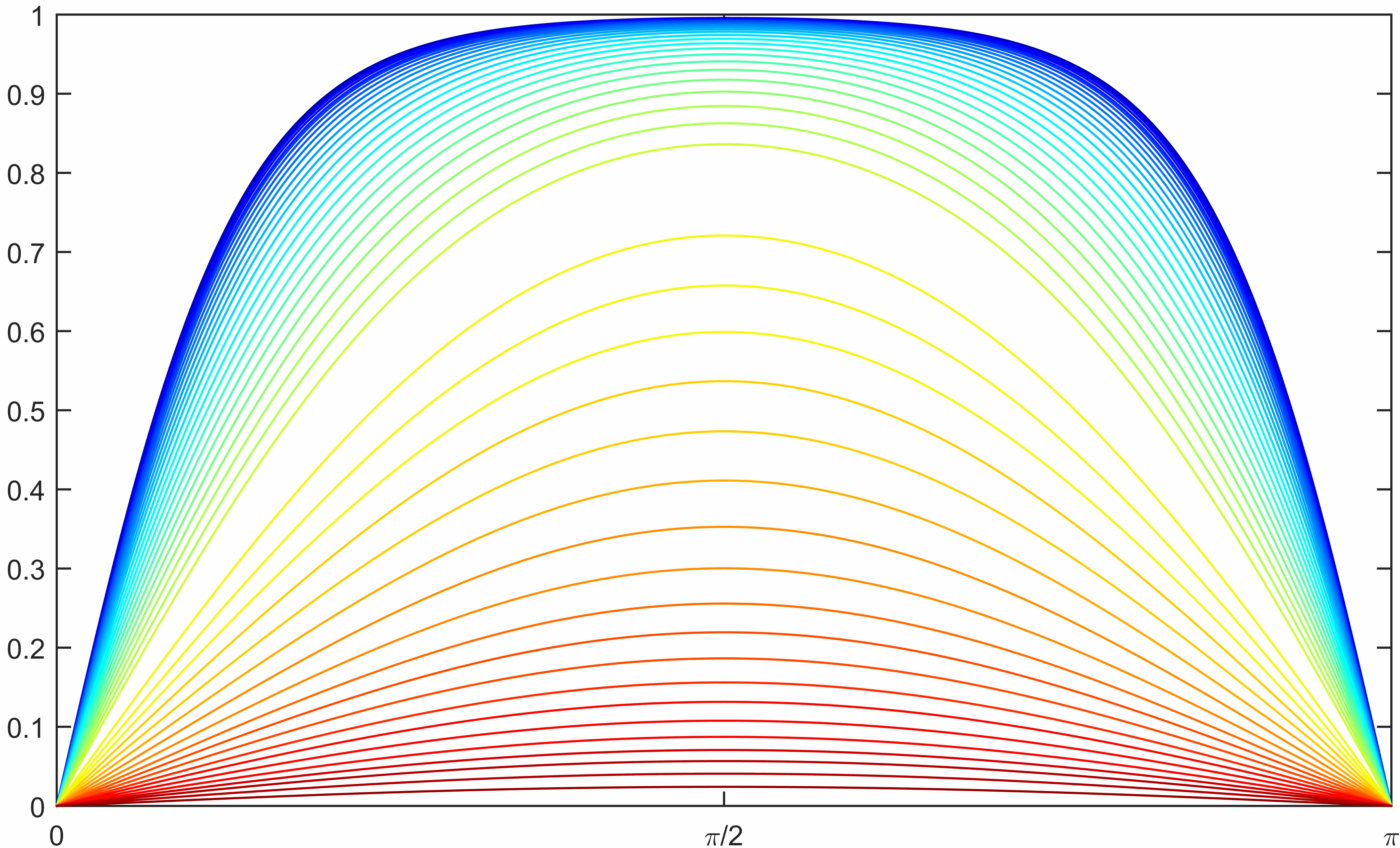}
     	\put(50,-8){$\Theta$}
       	\put(-8,32){\begin{turn}{90}$f$\end{turn}}
  \end{overpic}
        \caption{\phantom{Enough empty space to clear room}}
        \label{fig:1e}
  \par\medskip 
  \begin{overpic}[width=0.8\textwidth]{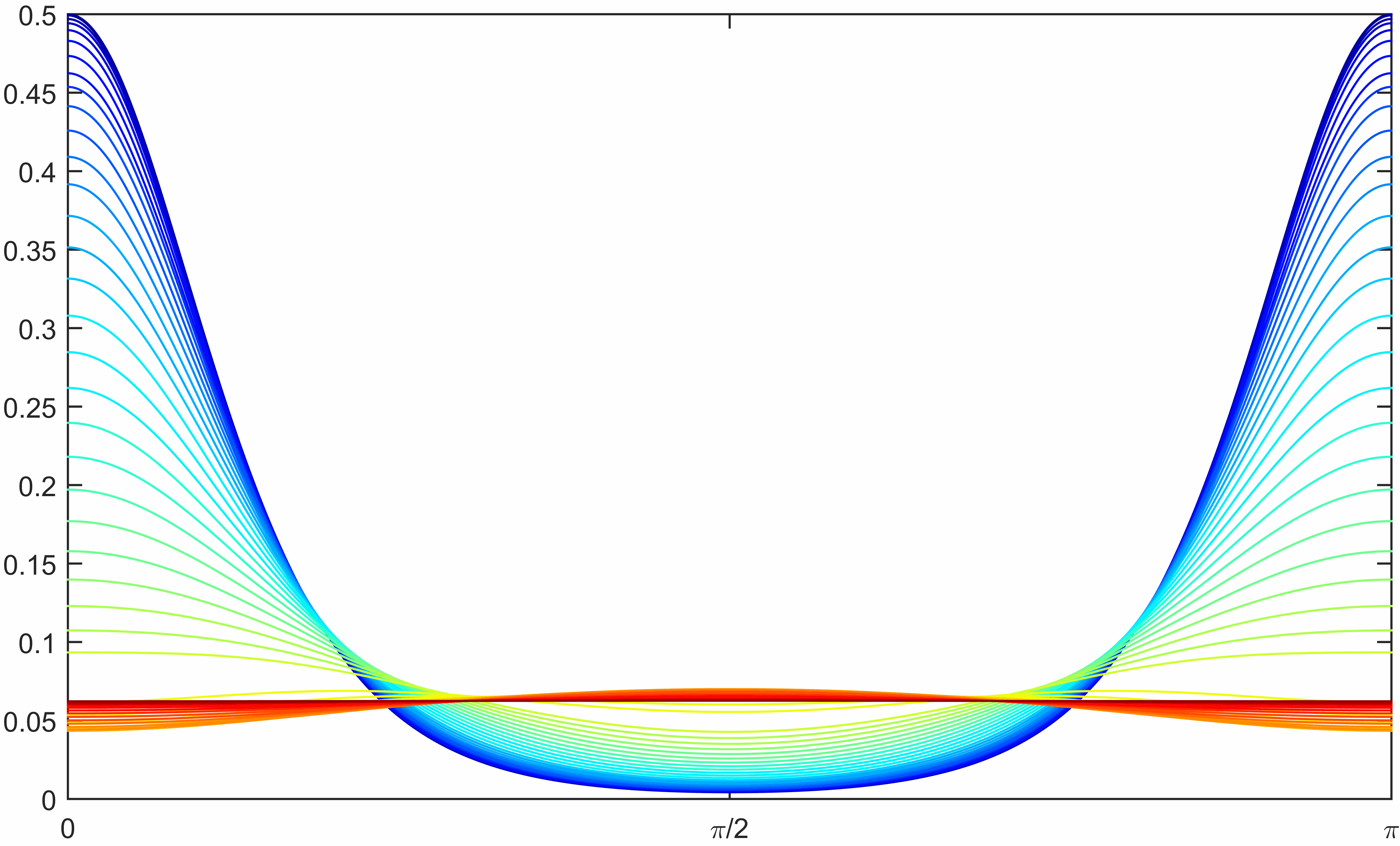}
     	\put(50,-8){$\Theta$}
       	\put(-8,28){\begin{turn}{90}$\star B$\end{turn}}
  \end{overpic}
        \caption{\phantom{Enough empty space to clear room}}
        \label{fig:1f}
  \par\medskip 
  \begin{overpic}[width=0.8\textwidth]{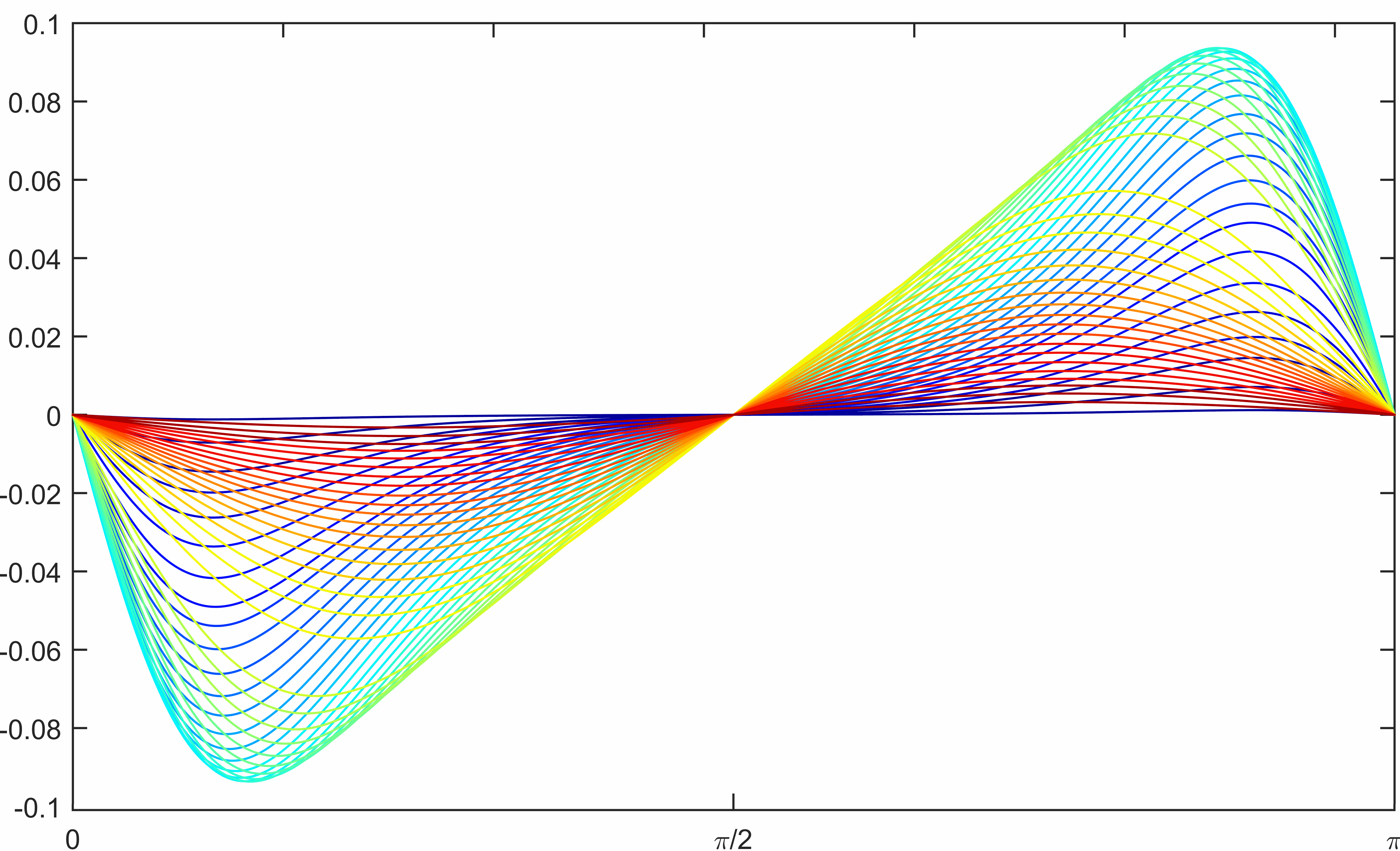}
     	\put(50,-8){$\Theta$}
       	\put(-8,18){\begin{turn}{90}$\left\langle e, R\d\Theta\right\rangle$\end{turn}}
  \end{overpic}
        \caption{\phantom{Enough empty space to clear room}}
        \label{fig:1g}
  \par\medskip 
  \begin{overpic}[width=0.8\textwidth]{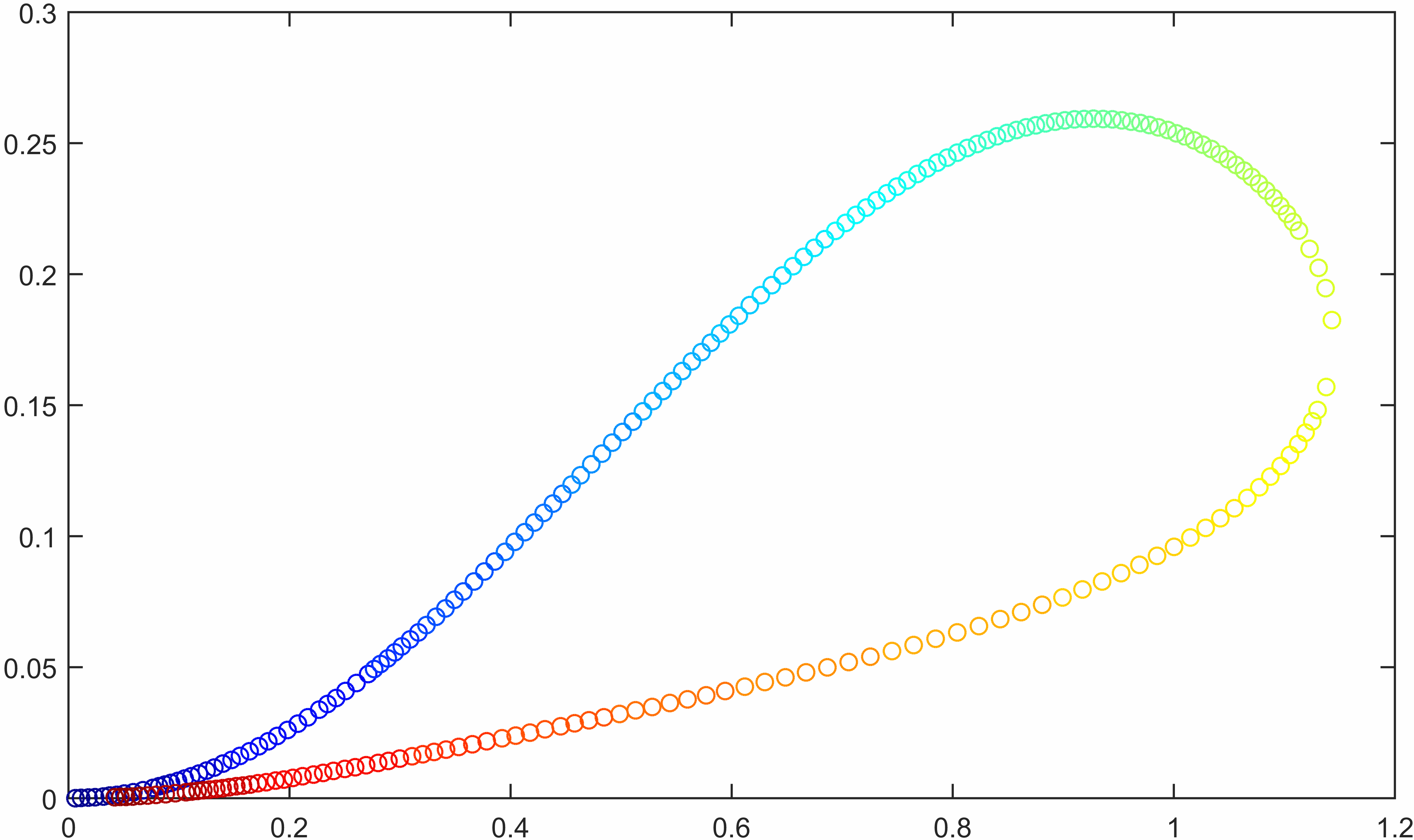}
     	\put(50,-8){$\kappa$}
       	\put(-8,24){\begin{turn}{90}$E_{elec}$\end{turn}}
  \end{overpic}
        \caption{\phantom{Enough empty space to clear room}}
        \label{fig:1h}
\end{subfigure}
        \caption{\small Snapshots along the curve of vortex solutions with $n_+=2$ and $n_-=0$ ((a)-(d)) and $n_+=1$ and $n_-=1$ ((e)-(h)) for a sphere of radius $R=4$ starting at the undeformed $\kappa=0$ solution.  (a), (e): Higgs profile function $f$, (b), (f): magnetic and (c), (g): electric fields each plotted against the sphere's angle of declination $\Theta$. The colours in (a)-(c) (resp. (e)-(g)) correspond to position along solution curves (d), (resp. (h)): the total electric energy versus $\kappa$. Note that the solution curves are not closed (as can be seen, for example, from the magnetic field plots (b), (f)).} 
        \label{fig:curve}
\end{figure}

\begin{figure}
\begin{subfigure}{0.5\textwidth}
  \centering
  \begin{overpic}[width=0.9\textwidth]{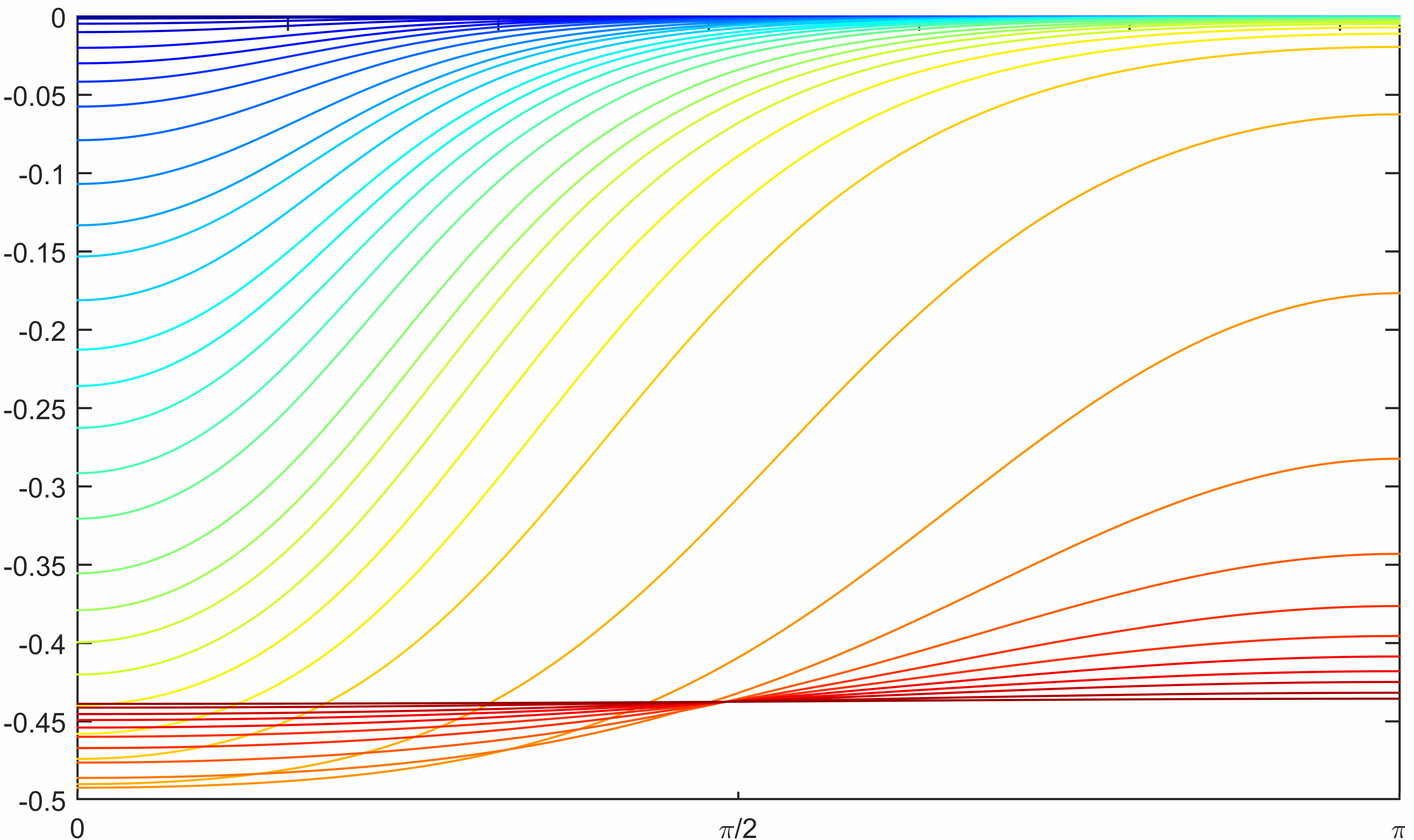}
   	\put(50,-5){$\Theta$}
   	\put(-6,27){\begin{turn}{90}$\kappa N$\end{turn}}
  \end{overpic}
  \caption{\phantom{Enough empty space to clear room}}
\end{subfigure}
\hfill
\begin{subfigure}{0.5\textwidth}
\centering
\begin{overpic}[width=0.9\textwidth]{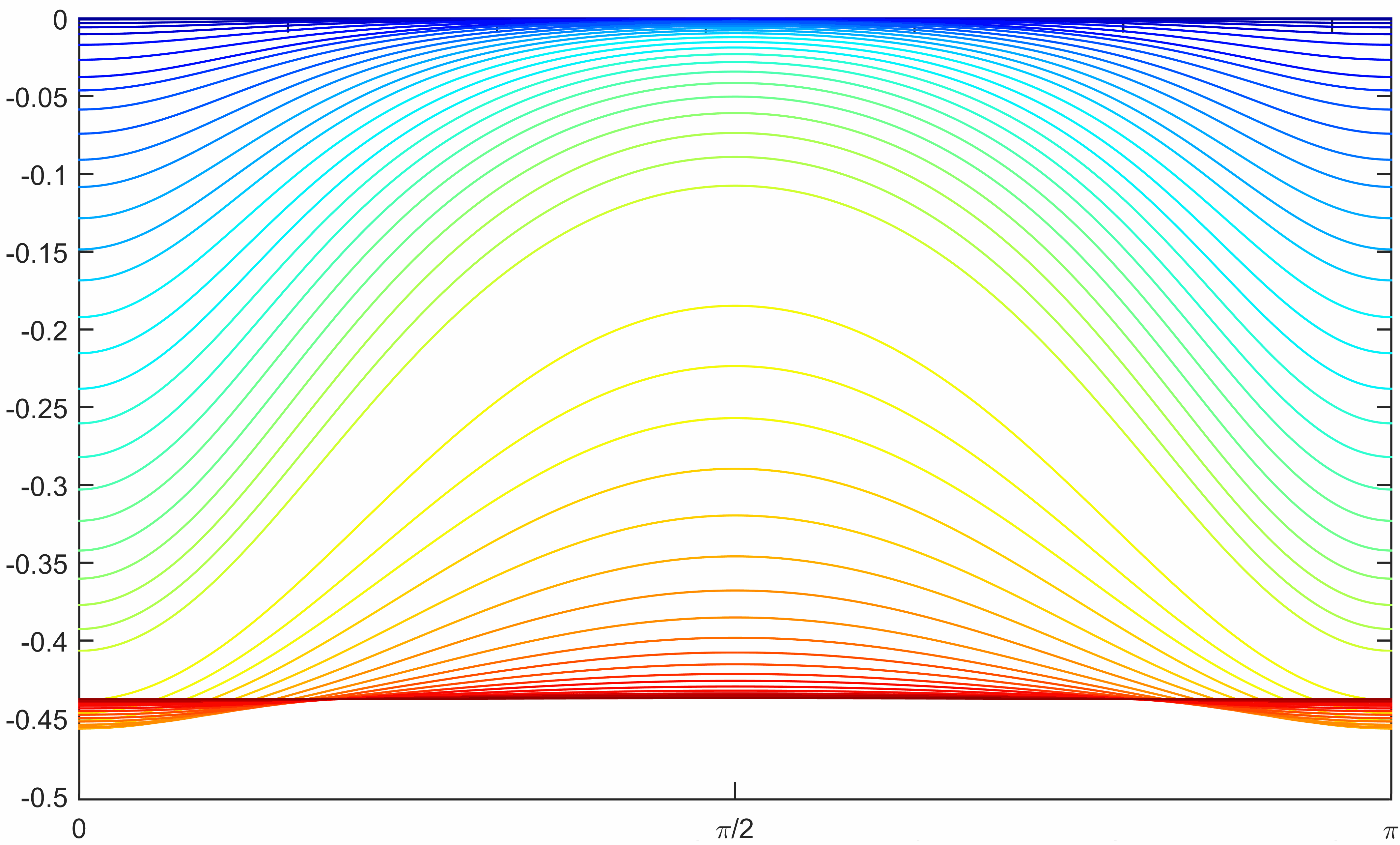}
  \put(50,-5){$\Theta$}
  \put(-6,27){\begin{turn}{90}$\kappa N$\end{turn}}
\end{overpic}
\caption{\phantom{Enough empty space to clear room}}
\end{subfigure}
\caption{\small Snapshots of $\kappa N$ against $\Theta$ for (a) $n_+=2$ and $n_-=0$ and (b) $n_+=1$ and $n_-=1$ with colouring to indicate the position along the vortex solution curve as in Figure \ref{fig:curve}(d) for (a) and Figure \ref{fig:curve}(h) for (b). As the curves (see Figure \ref{fig:curve}(d), (h)) reach $\kappa_*(D)$ and turn back to zero, $\kappa N \rightarrow \frac{n}{2R^2}-\frac12$, which for this configuration is $-\frac{7}{16}\approx 0.44$.}
\label{fig:kappaN}
\end{figure}

Since $\kappa(s)$ has a single turning point, at $s_*$ say, it attains a maximum value $\kappa_*(D)=\kappa(s_*)$.
Figure \ref{fig:kappamax} shows the maximal coupling $\kappa_*(D)$ for a range of $D$ and sphere radii $R$.
Note that $\kappa_*$ depends on both $n_+$ and $n_-$, not just their sum $n=n_++n_-$. In general, for fixed $n$,
the more evenly the divisor is split between the two poles (i.e.\ the smaller is $|n_+-n_-|$), the larger is
$\kappa_*(D)$. In fact, $\kappa_*(D)$ does not always depend monotonically on $n$: for instance for $R$ greater than $\approx$ 3.4, the maximal $\kappa$ for the symmetric $n = 4$ ($n_+ = n_- = 2$) configuration is larger than the maximal $\kappa$ for the asymmetric $n = 3$ ($n_+ = 3, n_- = 0$) configuration. Note also that in all cases $\kappa_*(D)$ is considerably smaller than the upper bound obtained in Theorem \ref{thm2}. The conjectured upper bound suggested by \cite{rictar}, equation
(\ref{rictarbounds}), is better, though still far from optimal, particularly at large $R$.

\begin{figure}
\centerline{\includegraphics[width=18cm]{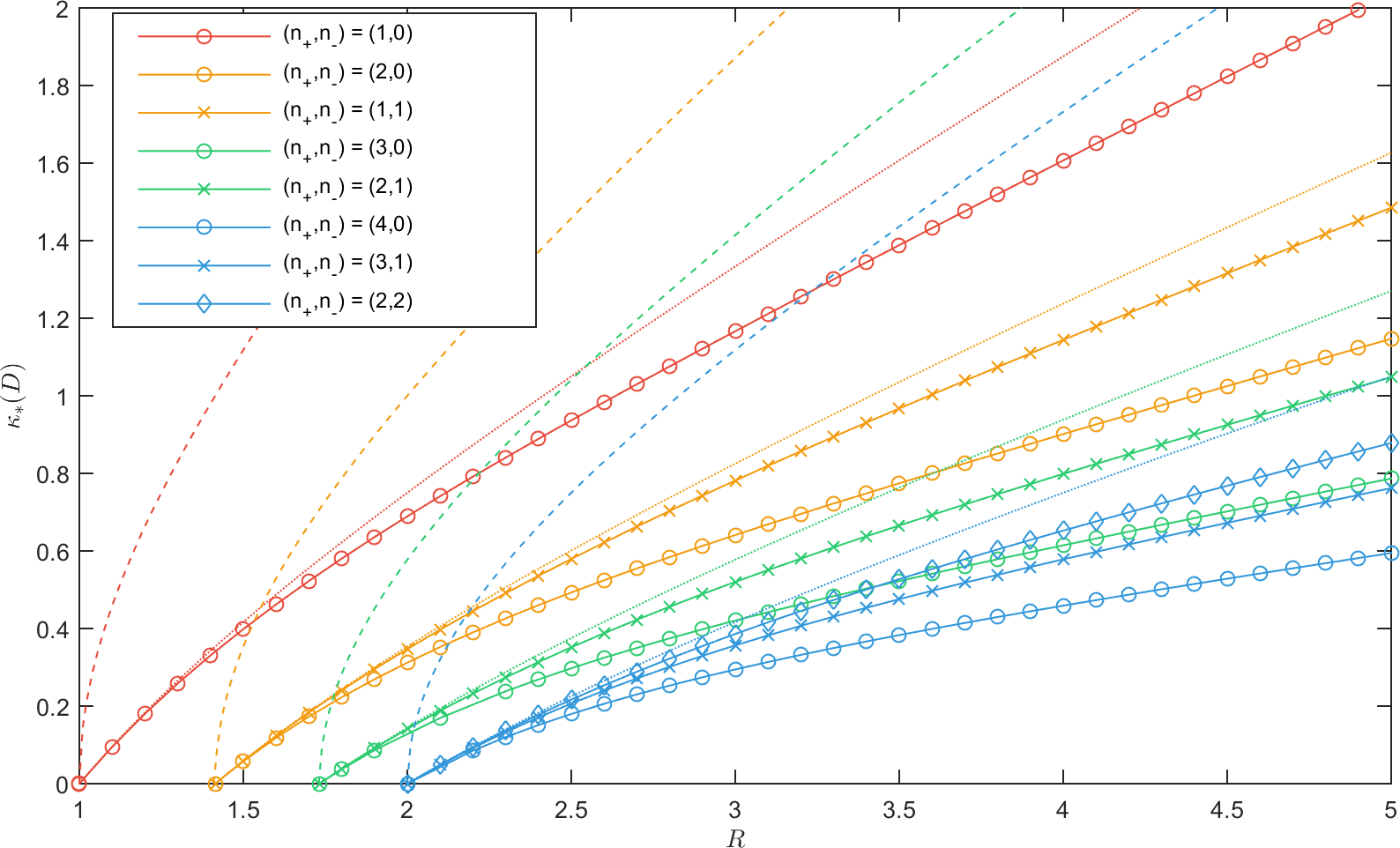}}
\caption{\small A plot of $\kappa_*(D)$, against sphere radius $R$ for various choices of divisor $D$. Data plotted in red, orange, green and blue corresponds to $n=n_++n_-=1,2,3,4$ respectively. The dashed lines represent our upper bound on $\kappa_*(D)$ (Theorem \ref{thm2}) and the dotted lines the bounds of Ricciardi and Tarantello \cite{rictar} (see \eqref{rictarbounds}). }
\label{fig:kappamax}
\end{figure}

\section{Higher dimensions}
\news

A key strength of the continuation/IFT strategy used here to prove existence of Maxwell-Chern-Simons vortices is its adaptability. In this section we illustrate this point by formulating, and proving vortex existence in, a higher dimensional generalization of the MCSH model. 

For the rest of this section, $\Sigma$ will be a compact K\"ahler manifold of complex dimension $k$, with K\"ahler form $\omega$. For any $p$-form $\xi$ we use $\xi^{[m]}$ to denote $\xi^m/m!$. 
Note that the volume form on $\Sigma$ is $\omega^{[k]}$, that
$*\omega^{[m]}=\omega^{[k-m]}$ and $|\omega^{[m]}|^2=k!/(m!(k-m)!)$ pointwise. 
We denote by $\Lambda:\Omega^p(\Sigma)\ra\Omega^{p-2}(\Sigma)$ the $L^2$ adjoint of the Lefschetz map
$L:\xi\mapsto\omega\wedge\xi$ (explicitly, $\Lambda=(-1)^p*L*$)
\cite[p.\ 139]{voi}. Other than this, 
we adopt the same notation as in section \ref{sec:MCSH}. It is convenient to define $\wt\L=
\pr^*\L$, where $\pr:I\times\Sigma\ra\Sigma$, $\pr(t,x)=x$, and $\wt{A}=\pr^*A(t)-iA_0(t)dt$ for the connexion on $\wt\L$ defined by the pair $(A(t),A_0(t))$. The curvature of this connexion is
\beq
\wt{F}=B(t)+dt\wedge e(t)
\eeq
where $B(t)=(\d A)(t)\in\Omega^2(\Sigma)$ and $e(t)=(\dot{A}-\d A_0)(t)\in\Omega^1(\Sigma)$ are time varying forms on $\Sigma$ interpreted as the magnetic and electric fields respectively. 

The field theory of interest has action
\beq
S=S_{\mbox{\tiny YMH}}+\kappa S_{\mbox{\tiny $\Psi$CS}}
\eeq
where $S_{\mbox{\tiny YMH}}$ is defined in (\ref{SYMH}) and 
\beq
S_{\mbox{\tiny $\Psi$CS}}=\frac12\int_{I\times\Sigma}\wt{A}\wedge\wt{F}\wedge\omega^{[k-1]},
\eeq
which we call the {\em pseudo-Chern-Simons} functional, since it coincides with $S_{\mbox{\tiny CS}}=\frac12\int
\wt{A}\wedge \wt{F}^k$ only in the case 
that $\Sigma$ has complex dimension $k=1$. As we will see, $S_{\mbox{\tiny $\Psi$CS}}$ defines a much more satisfactory electrodynamics for $k\geq 2$ than $S_{\mbox{\tiny CS}}$. 

The Euler-Lagrange equations for $S$ are (\ref{EL1}), (\ref{EL2}) along with
\bea
\dot{e}+\dstar B-\kappa*(e\wedge\omega^{[k-1]})&=&j\label{EL3-k}\\
-\delta e +\kappa \Lambda B&=&\rho\label{gauss-k}.
\eea
Note that these reduce to (\ref{EL3}) and (\ref{gauss}) when $k=1$, and maintain the desirable feature of being linear in electromagnetic fields.
This feature is broken, for $k\geq 2$, if we use $S_{CS}$ rather than $S_{\Psi CS}$. More importantly,
solutions of this system conserve the total energy $E$, defined in (\ref{Edef}), and there is a topological lower bound on $E$, saturated by solutions of a coupled system of PDEs generalizing the
Bogomol'nyi equations (\ref{bog1}), (\ref{bog2}), (\ref{bog3}), as we next demonstrate. 

The argument follows closely the Bogomol'nyi bound for the undeformed (that is, $\kappa=0$) model
in general dimension,
as presented, for example, in \cite{bra}, and the $k=1$
argument presented in section \ref{sec:MCSH}, so we will be brief. It relies on the so called K\"ahler identities, valid for any
metric connexion on the hermitian line bundle $\L$ over $\Sigma$:
\bea
\ip{\Lambda B,|\phi|^2}&=&\|\cd_A\phi\|^2-\|\ol{\cd}_A\phi\|^2,\label{KI1}\\
|B|^2\omega^{[k]}&=&|\Lambda B|^2\omega^{[k]}-B\wedge B\wedge\omega^{[k-2]}+4|B^{0,2}|^2\omega^{[k]}
\label{KI2}.
\eea
So, for any (possibly time dependent) collection of fields $(\phi,N,A,A_0)$,
\bea
E&=&\frac12\left\|\Lambda B-\frac12(1+2\kappa N-|\phi|^2)\right\|^2+2\|B^{0,2}\|^2+\|\ol{\cd}_A\phi\|^2
+\frac12\|D_0\phi+iN\phi\|^2+\frac12\|e+\d N\|^2\nonumber \\
&&+\frac12\|\dot{N}\|^2+
\ip{N,-\dstar e-\rho+\kappa\Lambda B}
+\frac12\int_\Sigma B\wedge\omega^{[k-1]}-\frac12\int_\Sigma B\wedge B\wedge \omega^{[k-2]},
\eea
and it follows that, for all fields satisfying Gauss's law (\ref{gauss-k}),
\beq\label{bogbound-k}
E\geq E_0(\L,[\omega]):=\frac12\int_\Sigma B\wedge\omega^{[k-1]}-\frac12\int_\Sigma B\wedge B\wedge \omega^{[k-2]},
\eeq
with equality if, and only if, in the gauge $A_0=N$, all fields are static and
\bea
\ol{\cd}_A\phi&=&0,\label{Bk1}\\
B^{0,2}&=&0,\label{Bk2}\\
\Lambda B&=&\frac12(1+2\kappa N-|\phi|^2),\label{Bk3}\\
\Delta N+\kappa\Lambda B&=&-N|\phi|^2.\label{Bk4}
\eea
Note that the quantity $E_0(\L,[\omega])$ on the right hand side of (\ref{bogbound-k}) depends only on the cohomology class of $B/2\pi$, which is a topological invariant of the bundle $\L$ (its first Chern class $c_1(\L)$), 
and the cohomology class of
$\omega$.  Hence, any static triple $(\phi,A,N)$ satisfying (\ref{Bk1})-(\ref{Bk4}) minimizes $E$ in its homotopy class, among fields satisfying (\ref{gauss-k}). It is straightforward to verify that
such fields also satisfy the Euler-Lagrange equations. In the case $k=1$, (\ref{Bk2}) is vacuously true on dimensional grounds, and the system reduces to (\ref{bog1})-(\ref{bog3}). 

In the case
$\kappa=0$, (\ref{Bk4}) is solved by $N=0$, and (\ref{Bk1})-(\ref{Bk3}) reduce to the usual vortex equations on a K\"ahler manifold \cite{bra}. A necessary condition for existence of solutions of this system (obtained
by integrating (\ref{Bk3}) over $\Sigma$) is that 
\beq
\frac1{4\pi}\Vol(\Sigma)\geq C_1(\L,[\omega]):=\frac{1}{2\pi}\int_\Sigma B\wedge\omega^{[k-1]}.
\eeq
If $\Vol(\Sigma)>4\pi C_1(\L,[\omega])$, solutions exist and are (up to gauge) in one-to-one
correspondence with effective divisors $D$ representing the homology class Poincar\'e dual to
$c_1(\L)$ \cite{bra}. Just as for vortices on a Riemann surface, $D=\phi^{-1}(0)$, but now $D$ is
a collection of irreducible analytic hypersurfaces in $\Sigma$ \cite[pp.\ 128-139]{grihar}. Our aim is to prove an existence, local uniqueness and smoothness result for pseudo-Chern-Simons deformations of these vortices.

\begin{thm} Let 
$D$ be an effective divisor homologous to the Poincar\'e dual of $c_1(\L)$, and assume $\Vol(\Sigma)>4\pi C_1(\L,[\omega])$. Then there exist $\kappa_*>0$ and a smooth curve
$$
(\phi,A,N):(-\kappa_*,\kappa_*)\ra \Gamma(\L)\times\Conn(\L)\times\Omega^0,
$$
 unique up to gauge, such that, for all $\kappa\in(-\kappa_*,\kappa_*)$,
$(\phi(\kappa),A(\kappa),N(\kappa))$ satisfies the Bogomol'nyi equations (\ref{Bk1})--(\ref{Bk4}) and $\phi(\kappa)^{-1}(0)=D$.
\end{thm}

\begin{proof}
By the results of \cite{bra}, there exists a unique (up to gauge) pair $(\hat\phi,\hat{A})$ 
satisfying (\ref{Bk1}), (\ref{Bk2}) and (\ref{Bk3}) with $\kappa=0$, such that $\phi^{-1}(0)=D$.
Now any other smooth section vanishing on $D$ has a unique representative in its gauge orbit of the
form $\phi=e^u\hat\phi$. Let $A=\hat{A}+i(\cd u-\ol\cd u)$. Then $\ol{\cd}_A\phi=e^u\ol{\cd}_{\hat{A}}\hat\phi=0$ by construction, and 
\beq
B=\hat{B}-2i\cd\ol\cd u,
\eeq
so $B^{0,2}=\hat{B}^{0,2}=0$. Hence
$(\phi,A)$ automatically satisfies (\ref{Bk1}) and (\ref{Bk2}). Furthermore, $[\Lambda,\cd]=i\ol{\cd}^\dagger$, where $\dagger$ denotes $L^2$ adjoint \cite[p.\ 139]{voi}, so
\bea
\Lambda B&=&\Lambda\hat{B}-2i([\Lambda,\cd]\ol{\cd}u+\cd\Lambda\ol{\cd}u)
=\Lambda\hat{B}+2\ol\cd^\dagger \ol\cd u+0
=\Lambda\hat{B}+2(\ol\cd^\dagger\ol\cd u+\ol\cd\, \ol\cd^\dagger u)\nonumber \\
&=&\frac12(1-|\hat\phi|^2)+\Delta u
\eea
since $\Lambda\ol\cd u=0=\cd^\dagger u$ on dimension grounds, and $\Delta=2\Delta_{\ol{\cd}}$
\cite[p.\ 141]{voi}. Hence, $(\phi,A,N)$ satisfies the remaining Bogomol'nyi equations (\ref{Bk3}), (\ref{Bk4}) if and only if $u,N$ satisfy (\ref{ariel1}), (\ref{ariel2}), with $f=|\hat{\phi}|^2$. The claim now follows immediately from Lemma \ref{natgra}.
\end{proof}

\begin{remark} Rehearsing the proof of Theorem \ref{thm2} with equations (\ref{Bk3}), (\ref{Bk4})
playing the roles of (\ref{bog2}), (\ref{bog3}), one easily deduces an upper bound on $\kappa_*(D)$
independent of $D$. Namely, if (\ref{Bk3}), (\ref{Bk4}) admit a smooth solution, then
\beq
\kappa^2\leq\frac{\Vol(\Sigma)}{4\pi C_1(\L,[\omega])}-1.
\eeq
Establishing a nontrivial lower bound on $\kappa_*(D)$ is more challenging.
\end{remark}

\section{Concluding remarks}\news

In this paper we used a simple Inverse Function Theorem argument to prove existence and local uniqueness of vortex solutions of the MCSH model in the small $|\kappa|$ regime on an arbitrary compact Riemann surface. This method has several advantages over more direct analytic approaches: it works uniformly on all geometries, immediately gives very strong smoothness information (both spatially and in terms of $\kappa$ dependence), and readily adapts to a natural high dimensional generalization of the model. We found a simple upper bound on $\kappa_*$, the maximal $|\kappa|$ for which vortices exist. Interpreting the solution curve $(\phi,A,N)(\kappa)$ as a solution of an (infinite dimensional) ODE problem, we proved existence of a positive {\em lower} bound on 
$\kappa_*(D)$, independent of the set of vortex positions $D$, the first result of this type that we are aware of.
It follows that the entire moduli space $M_n$ of $n$-vortex solutions persists for sufficiently small $|\kappa|$. We conducted a numerical study of rotationally equivariant vortices on round two-spheres, demonstrating that $\kappa_*(D)$ does, in fact, depend on the vortex positions, not just their number, and that the two distinct vortex solutions (suggested by
previous analysis of the model on flat tori \cite{rictar}) actually merge at $\kappa=\kappa_*(D)$. It would be interesting to see whether this is a generic phenomenon on compact domains.

Two other interesting questions suggest themselves. Can one develop a moduli space approximation to the low energy dynamics of vortices in this system? There are at least two competing conjectures for such a dynamics
\cite{colton,kimlee}, structurally similar (geodesic motion on $M_n$ perturbed by an effective magnetic field) but
known to differ from one another \cite{alqspe}. Neither has been rigorously derived from the parent field theory. Having proved persistence of the whole moduli space, this question is now well founded in the analytically simpler setting of compact domains where one can hope to make rigorous progress more easily.  

Second, can the IFT method pursued here be adapted to deal with Chern-Simons deformation of the
$O(3)$ sigma model, in which the Higgs field $\phi$ takes values in $S^2$? The undeformed model is known to support BPS vortex-antivortex superpositions 
\cite[ch.\ 11]{yan}. By shifting the vacuum manifold from the equator of the target $S^2$, one obtains a model where vortices and antivortices differ in size and mass, yet still coexist in marginally stable equilibrium. These have been proved to persist in the deformed model on $\R^2$ for small $|\kappa|$ in general \cite{hannam}, for all $\kappa$ in the case of pure vortex (or pure antivortex) solutions \cite{hanson}, and to be unique
in the case of pure coincident vortices (or antivortices) \cite{yancheche}. As for the MCSH model,
the model on compact domains seems to have been explored only in the special case of flat tori, see \cite{chiric}
for example. For mixed solutions (that is, those with both vortices and antivortices), the moduli space, even on compact $\Sigma$, is noncompact due to missing points where vortices and antivortices coalesce, so the global persistence of entire moduli spaces becomes an interesting and quite nontrivial problem.

\subsection*{Acknowledgements} JMS would like to thank Jonathan Partington and Alex Strohmaier for useful conversations about functional analysis, and SPF thanks Priya Subramanian for introducing him to pseudo-arclength continuation. SPF is an EPSRC Doctoral Prize Fellow.

\bibliographystyle{stylefile}
\bibliography{bibliography}

\end{document}